\documentclass[12pt]{article}

\usepackage[margin=1in]
{geometry} 
\usepackage{amsmath,amssymb,amsfonts}
\usepackage[mathlines]{lineno}

\usepackage{float}
\usepackage{svg}
\usepackage{xcolor}
\usepackage{subcaption}
\usepackage{algorithmic}
\usepackage{amsthm}
\usepackage{graphicx} 
\usepackage{amsmath} 
\usepackage{epsf,hyperref}
\usepackage[noadjust]{cite} 
\newtheorem{theorem}{Theorem}[section]

\newtheorem{lemma}[theorem]{Lemma}
\title{\textbf{SGRDN-Data learned sparsification of graph reaction-diffusion networks}}

\author{
Abhishek Ajayakumar\thanks{Email: abhishekajay@iisc.ac.in} \and 
Soumyendu Raha\thanks{Email: raha@iisc.ac.in}
}
\begin{document}
\maketitle
\begin{abstract}
Graph sparsification is an area of interest in computer science and applied mathematics. Sparsification of a graph, in general, aims to reduce the number of edges in the network while preserving specific properties of the graph, like cuts and subgraph counts. Computing the sparsest cuts of a graph is known to be NP-hard, and sparsification routines exist for generating linear-sized sparsifiers in almost quadratic running time $O(n^{2 + \epsilon})$. Consequently, obtaining a sparsifier can be a computationally demanding task, and the complexity varies based on the level of sparsity required. \textcolor{black}{We propose \texttt{SGRDN} to extend sparsification to complex reaction-diffusion systems. This approach seeks to sparsify the graph such that the inherent reaction-diffusion dynamics are strictly preserved on the resulting structure.} By selectively considering a subset of trajectories, we frame the network sparsification issue as a data assimilation problem within a Reduced Order Model (ROM) space, imposing constraints to conserve the eigenmodes of the Laplacian matrix ($L = D - A$), the difference between the degree matrix ($D$) and the adjacency matrix ($A$) despite perturbations. We derive computationally efficient eigenvalue and eigenvector approximations for perturbed Laplacian matrices and integrate these as spectral preservation constraints in the optimization problem. To further validate the method's broad applicability, we conducted an additional experiment on Neural Ordinary Differential Equations (Neural ODEs), where \texttt{SGRDN} successfully achieved parameter sparsity. 
\end{abstract}

\section{Introduction}
\label{sec:introduction}

 Modern deep learning frameworks, which utilize recurrent neural network decoders and convolutional neural networks, are characterized by a significant number of parameters. Finding redundant edges in such networks and rescaling the weights can prove very useful; one such work is done as in \cite{https://doi.org/10.48550/arxiv.2211.04598}. Another line of thought involves learning differential equations from data (\cite{Raissi_2018, Raissi_2018a, lu_data}), where the data generation part is rendered accessible by using sparse neural networks. These types of work find applications in areas like fluid dynamics, see \cite{fluidflow_arxiv}. Reducing the number of parameters in neural ODEs (ODENets), which can be used to model many complex systems is another motivation for this work.
We define a graph $G = (V, E, w)$ where $V$, $E$ denote the set of vertices and edges of the graph and ${w}$ denotes the weight function, that is ${w}: V\times V \rightarrow \mathbb{R}$. A sparsifier of a graph aims at reducing the number of edges in the graph such that the new reweighted graph $\bar{G} = (V, \bar{E}, \bar{w})$ achieves a specific objective like preserving the quadratic form of the Laplacian matrix, cuts, subgraph counts. E.g., In the case of spectral sparsifiers, we have
\[
\begin{split}
\frac{x^T L x}{(1+\epsilon)} \leq \, x^T \bar{L} x \leq \,(1 + \epsilon) x^T L x \;\; \forall \,x \in \mathbb{R}^n,\;\; \\   \textrm{where}\; \vert V \vert = n,  \epsilon \in (0,1). \\
\end{split}
\]
The notion of cut sparsification, which produces weighted subgraphs for which each cut is approximated, was introduced by Karger (\cite{karger_1994, karger_1999}). The idea involves sampling individual edges of the graph with probability $p$; the resulting sparsifier has size $O(pm)$, where $m = |E|$, the cardinality of the set of edges of the graph. An improved work for producing cut sparsifiers was based on the non-uniform sampling of edges as shown in (\cite{karger_2015, karger_1996}). There are other types of sparsifiers which preserve subgraph counts, described as in \cite{ahn_guha_mcgregor_2012}.
  
\par
Spielman and Teng introduced a more advanced form of sparsification known as spectral sparsifiers (\cite{spielman_teng_2011}). Spectral sparsifiers approximate the quadratic form of the graph Laplacian matrix. Such sparsifiers help in solving systems of linear equations (\cite{spielman_teng_2004, spielman_teng_2011}). The resulting sparsified graph exhibits an edge count approximated as $O(n \log(n)^c$), where $c$ is a large constant. In their study, Spielman and Srivastava (\cite{spielman_srivastava_2008}) demonstrated that it is possible to construct a sparsifier for any graph with $n$ vertices, containing $O(\frac{n \log(n)}{\epsilon^2}$) edges. Additionally, they introduced a linear algorithm to achieve this objective. Another work by Goyal, Rademacher and Vempala (\cite{goyal_rademacher_vempala_2009}) proposed to create cut sparsifiers from random spanning trees. Batson, Spielman and Srivastava (\cite{batson_spielman_srivastava_2014}) construct sparsifiers with size $O(\frac{n}{\epsilon^2}$) edges in $O(\frac{m n^3}{\epsilon^2 }$) time using a deterministic algorithm. Fung and Harvey devised an algorithm that samples edges based on the inverse of edge connectivity to produce cut sparsifiers (\cite{DBLP:journals/corr/abs-1005-0265}). Additionally they also describe an approach to produce cut sparsifiers from uniformly random spanning trees. A meta-learning-based approach for graph sparsification is discussed in (\cite{metalearning_sparsify}). \textcolor{black}{While the literature presents various graph sparsification techniques for different objectives, methods that preserve the dynamical behaviour of the system during sparsification remain unexplored.} 
\par
In this section, we review literature that concentrates on conserving the dynamical system described on a graph. The paper by Cencetti, Clusella, and Fanelli (\cite{cencetti_clusella_fanelli_2018}) explores methods for altering the network's topology while maintaining its dynamic characteristics intact; the new network obtained was isodynamic to the former. Kumar, Ajayakumar and Raha, in their work (\cite{KUMAR2021100948}), provide mathematical conditions for finding stable bi-partitions of a complex network and also provide a heuristic algorithm for the same. Graph partitioning does not alter the edge weights of graphs; however, it can cause the deletion of edges. This work shows that constrained internal patch dynamics with sufficient network connectivity ensure stable dynamics in subnetworks around its co-existential equilibrium point. 
\par The contributions of our article can be summarized as follows:
\begin{enumerate} \item In Section \ref{costfunc}, we introduce the \texttt{SGRDN} algorithm, which is used to sparsify reaction-diffusion dynamical systems on graphs. Towards achieving this goal: \begin{itemize} \item We propose estimates for eigenvalues and eigenvectors of the Laplacian matrix in Section \ref{eigen_approx}. \item The advantages of using such estimates are expounded in this section. \item The complexity of the procedure is described in Section \ref{timecomplexity} and empirically verified in Figures \ref{costfunccomplexity}, \ref{nonlinearcomplexity}. \end{itemize} \item In Theorem \ref{theorem_estimate}, we discuss bounds on the estimates proposed for the approximation of eigenvalues and eigenvectors of the Laplacian matrix under perturbations. \begin{itemize} \item Table \ref{tab:my_label} demonstrates how these estimates work well in practice. \end{itemize} \item In Section \ref{results}, we conduct extensive experimentation of the \texttt{SGRDN} algorithm on several graph datasets. \begin{itemize} \item The sparsity levels obtained on various graph datasets are shown in Table \ref{Allresults}. \item In Section \ref{odenetsection}, we demonstrate how the \texttt{SGRDN} algorithm can be adapted for generating sparse ODENets. \end{itemize} \end{enumerate}
\par
The structure of the article is as follows: we begin by presenting background material from spectral graph theory. A brief overview of reaction-diffusion systems on graphs is provided in the next section. We briefly introduce the adjoint method for data assimilation, which can be applied to the parameter estimation problem. In the subsequent sections, we present the reduced vector field for the chemical Brusselator reaction-diffusion system on the graph and discuss various performance metrics. Challenges involved in ensuring isodynamic behaviour are presented in the next section; the challenging aspect is obtaining good estimates of the eigenpair of the perturbed Laplacian matrix iteratively. We explore approximations of the eigenpair and their effectiveness in the next section, followed by formulating the optimization problem and its challenges. We present experimental results on real-world graphs to demonstrate the effectiveness of our approach. In Section \ref{odenetsection}, we show how to produce sparse neural ODENets using the framework for a linear dynamical system.

\section{Background}
\subsection{Laplacian matrix:($L$)}
The Laplacian of an undirected weighted graph $G = (V,E,w)$, where $|V| = n, |E| = m$ with a weight function $w:V\times V\rightarrow \mathbb{R}$ is given by
\[
  L \left(u,v\right)=\;\left\lbrace \begin{array}{ll}
d_v -w\left(u,v\right) & \mathrm{if}\;u=v,\\
-w\left(u,v\right) & \mathrm{if}\;u\;\mathrm{and}\;v\;\mathrm{are}\;\mathrm{adjacent},\\
0 & \mathrm{otherwise}\ldotp 
\end{array}\right.  
\]
Where $d_v = \sum_{u \sim v} w(u,v)$ denotes the degree of vertex $v$.
\subsection{Incidence matrix:($B$)}
The incidence matrix $B_{m \times n}$ with a given orientation of edges of a graph is given by
\[
B\left(e,v\right)=\;\left\lbrace \begin{array}{ll}
1 & \mathrm{if}\;v\;\mathrm{is}\;e^{\prime } s\;\mathrm{head},\\
-1 & \mathrm{if}\;v\;\mathrm{is}\;e^{\prime } s\;\mathrm{tail},\\
0 & \mathrm{otherwise}\ldotp 
\end{array}\right.    
\]
The graph Laplacian matrix for a given orientation with weights $\bar{w}_i = \bar{\gamma}_i w_i,\;i = 1,2,\ldots, m$ can be expressed as
\[
    \begin{split}
    \bar{L} = B^T W^{1/2} \textrm{diag}(\bar{\gamma}) W^{1/2} B \quad \textrm{, where } \bar{\gamma} \in \mathbb{R}^m  \textrm{ represents the multipliers of the weights}, \\    
    \end{split}
\label{multipliers_weight}
\]
\textrm{$W$ represents the diagonal matrix of weights $w_p,\; p = 1,2,\ldots,m.$}

\section{Dynamical systems involving graphs}
\label{main-text}
The general form of a reaction-diffusion system is outlined below, as presented in the work of Cencetti et al. (\cite{cencetti_clusella_fanelli_2018}). Each node in the network, denoted as $i$, possesses an $\mathrm{m}$-dimensional variable $r_i(t) \in \mathbb{R}^{\mathrm{m}}$ that characterizes its activity at time $t$. The evolution of $r_i$ over time follows a specific set of rules, beginning from an initial state $r_i(0)$ as described below,
\[ \frac{dr_i}{dt} = \mathcal{F}(r_i) + K \sum_{j=1}^{n} A_{ij} \mathcal{G}(r_j-r_i) \hspace{5mm} i = 1,2,....n\]
\textcolor{black}{Here, $\mathcal{F}$ denotes the reaction component and the remaining terms explain the diffusion in the graph with $A$ denoting the adjacency matrix of the graph.} $\mathcal{F} : \mathbb{R}^{\mathrm{m}} \rightarrow \mathbb{R}^{\mathrm{m}}$, $\mathcal{G}: \mathbb{R}^{\mathrm{m}} \rightarrow \mathbb{R}^{\mathrm{m}}$. One such reaction-diffusion system is given by the alternating self-dynamics Brusselator model~\cite{LANDSBERG1972, cencetti_clusella_fanelli_2018}, where we couple the inter-patch dynamics using the graph Laplacian matrix,
\begin{equation}
    \left\lbrace \begin{array}{l}
\dot{x}_i \;=\;a-\left(b+d\right)x_i +c\;x_i^2 y_i -D_x \;\sum_j L {\;}_{\mathrm{ij}} x_j \\
{\dot{y}_i \;=\;{b}}x_i -c\;x_{\;i}^2 y_i -D_y \;\sum_j L_{\mathrm{ij}} \;y_j \;\;\;\;\;\;\;
\end{array}\right.
\label{Rd dynamics}
\end{equation}
\textcolor{black}{
With $r_i = (x_i, y_i) \in \mathbb{R}^2$, the system defined by Equation~\eqref{Rd dynamics} possesses fixed points as explained in~\cite{cencetti_clusella_fanelli_2018}:
\[r^{\star} = (x_i,y_i) = \left(\frac{a}{d}, \frac{bd}{ac}\right) \quad \forall \, i.\]}
\textcolor{black}{Our primary objective is to find a graph $\bar{G} = (V,\bar{E}, \bar{w})$ with a new sparse weight function $\bar{w}$ such that on solving Equation (\ref{Rd dynamics}) using the updated graph $\bar{G}$ we get solutions $\bar{x_i} \in B(x_i,\delta_i), \bar{y_i} \in  B(y_i, \Delta_i) \;\forall \;i$, where $B(z,r) = \{w : \|w - z\| \leq r\}$ denotes a closed ball of radius $r$ centered at $z$, and $\delta_i$ and $\Delta_i$ are relatively small.} We will demonstrate the problem on a linear dynamical system involving the Laplacian matrix. The heat kernel matrix involving the Laplacian matrix is given as follows, \\
\begin{equation}
H_t  = e^{-Lt}    
\label{kernel}    
\end{equation}
\\
\textcolor{black}{
Differentiating Equation (\ref{kernel}) with respect to time, we get the matrix differential equation \[\frac{dH_t}{dt} = -L H_t, \;\; H_t \in \mathbb{R}^{n\times n}.\]
This system can be considered as a system of ODEs of the form
\begin{equation}
    \frac{d\mathcal{F}(i,t)}{dt} = -L \mathcal{F}(i,t), \quad \mathcal{F} \in \mathbb{R}^n.
    \label{Linear_ode}
\end{equation}
We denote the solution at node $i$ and time $t$ by $\mathcal{F}(i,t)$. According to the fundamental theorem for systems of linear first-order ODEs, we know that the solution to such a system with initial condition $\mathcal{F}(i,0) = f(i)$ is $\mathcal{F}(i,t) = e^{-Lt} f(i)$ and this solution is unique.} The underlying quantity which follows ODE described in Equation (\ref{Linear_ode}) could be temperature, chemical reactants, and ecological species. If we use an approximation for $L$, let it be $\bar{L}$ in the above equation with the same initial condition,
the solution now will be solutions to the ode of the form
\begin{equation}
\hspace{25mm}\frac{\partial \bar{\mathcal{F}}}{\partial t} = -\bar{L} \bar{\mathcal{F}} \label{lineardiffusion}   
\end{equation}

$\bar{\mathcal{F}}(i,t) = e^{-\bar{L}t}f(i)$. The error in this substitution at node $i$ for a time $t$ is given by \[Q(i,t) = \mathcal{F}(i,t) - \bar{\mathcal{F}}(i,t),\] \[Q(i,t) = (\sum_{k=1}^{\infty} \frac{{L}^k (-t)^k}{k!})f(i) - (\sum_{k=1}^{\infty} \frac{{\bar{L}}^k (-t)^k}{k!})f(i) .\]
\par Given random initial conditions \{$\mathcal{F}_{0\alpha}\} \in \mathbb{R}^n, \alpha = 1,2 ....,\omega$, one could discretize system as described by Equation (\ref{lineardiffusion}) to obtain solutions at various time points using a numerical scheme. The discretization of the dynamics for the system described by Equation (\ref{lineardiffusion}) is as follows,

\begin{equation}\hspace{5mm}
\bar{\mathcal{F}}^{\alpha}_{q+1} = M(\bar{\mathcal{F}}^{\alpha}_{q},\bar{w})    
\label{forward_eq}
\end{equation} 
where $M:\mathbb{R}^{n} \times \mathbb{R}^m \rightarrow \mathbb{R}^{n}$ 
with $M = (M_{1},M_{2},.....,M_{n})^T, M_{i} = M_{i}(\bar{\mathcal{F}}^{\alpha}_q,\bar{w})$ for $1\leq i \leq n$ and $ 0\leq q \leq S $, $S$ denotes the last time-step, and $\bar{\mathcal{F}}^{\alpha}_{0}  \in \mathbb{R}^n$ 
is the initial condition which is assumed to be known. 
The primary objective would be to obtain a sparse vector of weights $\bar{w}$ which minimizes the error $Q$ as explained above, subject to the discrete dynamics shown in Equation (\ref{forward_eq}). 
This forms the core of the data assimilation problem. 
To know more about data assimilation, one could refer to (\cite{lewis_lakshmivarahan_dhall_2009},
\cite{sarkka}, \cite{law_stuart_zygalakis_2015}).
\section{Adjoint method for data assimilation}
Parameter estimation problems have rich literature (\cite{1p}, 
\cite{caracotsios_stewart_1985},
\cite{3p},
\cite{4p}). 
Using Unscented and Extended Kalman filtering-based approaches is not recommended when constraints on parameters are imposed (\cite{EKf_mhe}). We use the Adjoint sensitivity method for the parameter estimation (\cite{lakshmivarahan_lewis_2010}). In this section, we provide a brief overview of this method and how it could be used in parameter estimation. \\
\small\textbf{Statement of the Inverse Problem:} The objective is to estimate the optimal control parameter ${\alpha}$ that minimizes the cost function $J({\alpha})$ based on a given set of observations  $\{\textbf{$z_k$} \in \mathbb{R}^n \mid 1 \leq k \leq N\}$. These observations represent the true state, and the model state $x_k$ follows a specific relationship described by $x_k = M(x_{k-1}, {\alpha})$. The state variable $x$ belongs to $\mathbb{R}^n$, denoted as $x = (x_1, x_2, \dots, x_n)$, and $\mathcal{F}:\mathbb{R}^n \times \mathbb{R}^p \rightarrow \mathbb{R}^n$ represents a mapping function. Specifically, $\mathcal{F} = (\mathcal{F}_1, \mathcal{F}_2,\dots ,\mathcal{F}_n)$ where $\mathcal{F}_i = \mathcal{F}_i(x, {\alpha})$ for $1 \leq i \leq n$ and ${\alpha} \in$ $\mathbb{R}^p$. \newline
Consider the nonlinear system of the type
\begin{equation}
\frac{d x}{d t} = \mathcal{F}(x, {\alpha}),
\label{IVPprob}
\end{equation}
 with $x(0) = c,$ being the initial condition assumed to be known, ${\alpha}$ is a system parameter assumed to be unknown. The vector $x(t)$ denotes the state of the system at time $t$, and $\mathcal{F}(x, \alpha)$ represents the vector field at point $x$. If each component of $\mathcal{F}_i$ of $\mathcal{F}$ is continuously differentiable in $x$, the solution to the initial value problem (IVP) exists and is unique.
\newline \newline
The Equation (\ref{IVPprob}) can be discretized using many schemes (\cite{burden_faires_burden_2016}), and the resulting discrete version of the dynamics can be represented as
\begin{equation*}
x_{k+1} = M(x_k, \alpha)
\label{discretivp}
\end{equation*}

The function $M:\mathbb{R}^n \times \mathbb{R}^p \rightarrow \mathbb{R}^n$ can be represented as $M = (M_1, M_2, \ldots, M_n)$, where each $M_i = M_i(x_k, \alpha)$ for $1 \leq i \leq n$.  \newline \newline

Let us define a function $J:\mathbb{R}^p \rightarrow \mathbb{R}$ as follows (\cite{lewis_lakshmivarahan_dhall_2009})
\begin{equation*}
    J(\alpha) = \frac{1}{2}\sum_{k=1}^N (z_k - x_k)^T(z_k - x_k)
\end{equation*}
\subsection{Adjoint sensitivity with respect to parameter $\alpha$}
\label{adj_param}
\[
\begin{array}{l}
\delta J(\alpha) \;= \sum_{k=1}^N \eta^T_k \delta x_k \;\;   \;\text{\; for more details refer \cite{lakshmivarahan_lewis_2010}}\\ \\
\;\;\;\;\;\;\;\;\;\;\;\;= \sum_{k=1}^N \eta^T_k V_k \delta \alpha\\ \\
\;\;\;\;\;\;\;\;\;\;\;\;\;= \left(\sum_{k=1}^N V_k^T \eta_k\right)^T \delta \alpha\\ 
\end{array} 
\]
\begin{equation}
    \text{where \;} V_k = A_{k-1}V_{k-1} + B_{k-1}
\label{TLS}
\end{equation}
\textcolor{black}{\[A_{k-1} = D_{k-1}(M(x_{k-1},\alpha)) \in \mathbb{R}^{n\times n} \text{\;and\;} B_{k-1} = D_{\alpha}(M(x_{k-1},\alpha)) \in \mathbb{R}^{n\times p},\]
where $D_{k-1}(M(x_{k-1},\alpha))$ and $D_{\alpha}(M(x_{k-1},\alpha))$ denote the Jacobian matrices of $M(x_{k-1},\alpha)$ with respect to $x_{k-1}$ and $\alpha$, respectively.} The recursive relation focuses on finding the sensitivity of $J$ with respect to the parameter $\alpha$. Iterating Equation (\ref{TLS}), we get \[ \begin{array}{l}
V_k = \sum_{j = 0}^{k-1} (\prod_{s = j+1}^{k-1} A_s) B_j \\ 
\end{array}\]
The above recursive relation deals with the product of sensitivity matrices.   
\newline
The first variation of $J$ after specific calculations is given by (see \cite{lewis_lakshmivarahan_dhall_2009}) 
\begin{equation}
\delta J = \left(\sum_{k=1}^N  \sum_{j = 0}^{k-1}  B_j^T(\prod_{s = k-1}^{j+1} A_s^T) \eta_k\right)^T \, \delta\alpha\;
\label{first_var}
\text{, where}\;\;\eta_k = x_k - z_k, \\ \prod\,\text{denotes reverse product.}
\end{equation}

From first principles
\[\delta J = \nabla_{\alpha} J(\alpha)^T \delta \alpha\]
Comparing first principles with Equation (\ref{first_var}), we get \[\nabla_{\alpha} J(\alpha) = \sum_{k=1}^N  \sum_{j = 0}^{k-1}  B_j^T(\prod_{s = k-1}^{j+1} A_s^T) \eta_k  \]

\subsubsection{Adjoint method Algorithm}
\label{Adjointalgo}
\begin{enumerate}
    \item Start with $x_0$ and compute the nonlinear trajectory $\{x_k\}_{k=1}^N$ using the model $x_{k+1} = M(x_k, \alpha)$
    \item $\lambda_N = \eta_N$
    \item Compute $\lambda_j = A_j^T \lambda_{j+1} + \bar{\eta_j}$ for $j = N-1 \text{\;to\;} k$, $\bar{\lambda}_k = B_{k-1}^T \lambda_k$, $k = 1 \text{\;to\;} N$, $\bar{\eta}_j = \delta_{j, j_i} \eta_j$ with $\delta_{j,j_i} = 1$ when $j = j_i$ , else 0 
    \item sum = \underline{0}, sum = sum +  $\bar{\lambda}_k$ vectors from $k = 1 \text{\;to\;} N$
    \item $\nabla_{\alpha}J(\alpha) = \text{sum}$
    \item Using a minimization algorithm find the optimal $\alpha^*$ by repeating steps 1 to 5 until convergence
\end{enumerate}
\subsubsection{Computing the Jacobians $D_{k} \left(M\right)$ and $D_{\alpha} \left(M\right)$
}
\label{jacobians}
This section discusses how to compute the Jacobian matrices described in Section \ref{adj_param}. The numerical solution of an ODE at timestep $j+1$($y_{j+1} \approx x_{j+1}$)  with slope function $\frac{dx}{dt} = f(t,x)$ given $d$ slopes($k_d$), where $x \in \mathbb{R}^n$ and $f: \mathbb{R} \times \mathbb{R}^n \rightarrow \mathbb{R}^n$ is given by $y_{j+1} = y_{j} + b_1 k_{1} + b_2 k_2 + ....+ b_d k_d$, where $k_i = hf(t + c_ih, x + a_{i1}k_1 + a_{i2}k_2 + ... a_{id}k_d)$. See \cite{burden_faires_burden_2016} to learn more about numerical analysis.
\par We demonstrate computing $D_{k}(M)$ using an explicit numerical scheme, particularly the Euler forward method on the Lotka-Volterra model.

\[\left\lbrace \begin{array}{ll}
\begin{array}{l}
\frac{{\mathrm{dx}}_1 }{\mathrm{dt}\;}=\alpha x_1 - \beta x_1 x_2  \\ \\
\frac{{\mathrm{dx}}_2 }{\mathrm{dt}}=\delta x_1 x_2 - \gamma x_2 \\ \\
\end{array} & 
\end{array}\right.\]
When we apply the standard forward Euler scheme, we obtain
\[ x_{k+1} = M(x_k) \] \\
where $x_k = (x_{1k}, x_{2k}),\; \textbf{M}(x_k) = (M_1(x_k), M_2(x_k))$ where 
\[\begin{array}{l}
M_1 \left(x_k \right)=x_{1k} +\left(\Delta t\right) \left(\alpha\,x_{1k} -\beta\,x_{1k}x_{2k} \right) \\
M_2 \left(x_k \right)=x_{2k} +\left(\Delta t\right)(\delta\;x_{1k} x_{2k} -\gamma\,x_{2k}) 
\end{array} \]
It can be seen that 
\[
D_{k}(M) = A_k = \left\lbrack \begin{array}{cc}
1+\Delta t\;\left(\alpha -\beta \;x_{2k} \right) & -\Delta t\;\beta \;x_{1k} \\
\Delta t\;\delta \;x_{2k}  & 1+\Delta t\;\left(\delta \;x_{1k} -\gamma \;\right)
\end{array}\right\rbrack
\]

\[
D_{\alpha}(M) = B_k =  
\left\lbrack \begin{array}{cccc}
\Delta t x_{1k}  & \Delta t\left(-x_{1k} x_{2k} \right) & 0 & 0\\
0 & 0 & -\Delta tx_{2k}  & \Delta tx_{1k} {\;x}_{2k} 
\end{array}\right\rbrack   
\]

\subsection{Using reduced order model for dynamical systems}
\label{ROMappro}
When dealing with graphs of considerable size, the dimensionality of the state vector becomes a concern due to its potential impact on computational efficiency. To address this issue, dimensionality reduction techniques like POD, also known as Karhunen–Lo`{e}ve, can be applied. The work described in \cite{rathinam_petzold_2003} discusses the utilization of POD as a means to alleviate the computational complexity associated with large graph sizes. The procedure requires snapshots of solutions of a dynamical system with a vector field $f$ $\in \mathbb{R}^n$,
\begin{equation}\hspace{5mm}
\dot{x} = f(x,t).
\label{dynamicalPOD}
\end{equation}

Using the snapshot solutions, we could create a matrix $\rho$ of projection $\in \mathbb{R}^{k\times n}$ where $k$ denotes the reduced order dimension and $\bar{x}$ denotes the mean corresponding to snapshot solutions, see \cite{rathinam_petzold_2003} for more details. The reduced order model of the system as described by Equation (\ref{dynamicalPOD}) is then given by 

\[
     \dot{z} = \rho f(\rho^T z + \bar{x}, t).
\]
If we are solving an initial value problem with $x(0) =x_0$ , then
the reduced model will have the initial condition $z(0) = z_0$ , where
 
\[
z_0 = \rho(x_0 - \bar{x}).
\]
The reduced order model for the system as described by Equation (\ref{lineardiffusion}) given the matrix $\rho$ of projection is 
\[\dot{z} = -\rho \bar{L}\rho^T z - \rho \bar{L} \bar{x} = \psi(z, \bar{w}).\]
The reduced order model for the system described by Equation (\ref{Rd dynamics}) is obtained as follows. \\
\textbf{Notation:}
If $p$ and $q$ are vectors, then $p \,\odot \, q$ represents element-wise multiplication, and $p^{\circ n}$ raises every element of the vector $p$ to power $n$. Let $B$ be a matrix $\in \mathbb{R}^{n\times n}$, then $B(1:2,:)$ indicates the first 2 rows of $B$ and all the columns of the matrix $B$.  \\
Let multipliers be $\bar{\gamma}$, where $\bar{w} = W \bar{\gamma}$. $W$ represents the diagonal matrix of weights of the graph with weights $w_i, i = 1,2,\ldots, m$. $\bar{w}_i$ denotes the $i-th$ entry of the new weight, using which a new Laplacian matrix can be constructed. \\
Let \textcolor{black}{
\begin{align*}
V_1      &= \rho^T(1:n,:)z + \bar{x}(1:n), \\
V_2      &= \rho^T(n+1:2n,:)z + \bar{x}(n+1:2n), \\
V_3      &= \rho(:, 1:n), \\
V_4      &= \rho(:, n+1:2n), \\
V_5      &= B^T W^{1/2} \Big(\mathrm{diag}\big(W^{1/2}B V_3^T z + W^{1/2}B\bar{x}(1:n)\big)\Big), \\
V_6      &= B^T W^{1/2} \Big(\mathrm{diag}\big(W^{1/2}B V_4^T z + W^{1/2}B\bar{x}(n+1:2n)\big)\Big), \\
V_e      &= \rho^T(e,:)z + \bar{x}_e, \\
\hspace{10mm}V_{n+e}  &= \rho^T(n+e,:)z + \bar{x}_{n+e}, \\
V_f      &= \rho^T(f,:)z + \bar{x}_f, \\
\;\;\;V_{f-n}  &= \rho^T(f-n,:)z + \bar{x}_{f-n}.
\end{align*}
}

\begin{align}  
\label{fa}
 \dot{z}  \qquad & =  \mathbb{\psi}(z,\bar{\gamma})  \\
 & =  \rho \;h(z) -D_x \;V_3 \bar{L} \nonumber
  V_1  -D_y\; V_4 \bar{L}\nonumber
 V_2
\end{align}

\begin{align*}
h(z) = 
\begin{pmatrix} 
a \mathrm{1}_{n\times1} - (b+d)V_1  
+ c[V_1^{\circ 2} ] \odot V_2  \\ \\
bV_1 - c(V_1^{\circ 2})\odot V_2  
\end{pmatrix}   
\end{align*}

\begin{align*}
\hspace{10mm} &\nabla_{z} h_e{\left(z\right)}^T \bigg\vert_{e=1,2,..,n} &&= -\left(b+d\right)\rho^T \left(e,:\right) + c \left(V_e^2 \right)\rho^T \left(n+e,:\right) \
+  2c V_e V_{n+e}\;\rho^T(e,:) 
\\
& \nabla_z {h_f(z)}^T \bigg\vert_{f = n+1,n+2,..,2n} && = b \rho^T(f-n,:)  - c(V_{f-n})^2 (\rho^T(f,:)) - \,2c V_{f-n}\,V_f\,\rho^T(f-n,:) \\
 & D_{z}\psi(z,\bar{\gamma}) &&= \rho D_z h(z) -D_x \, V_3 \bar{L} V^T_3  -D_y\,V_4 \bar{L} V^T_4  \\
& D_{\bar{\gamma} } \psi (z,\bar{\gamma} ) &&= -D_x V_3 V_5  -  D_y \;V_4 V_6
\end{align*}

\subsection{POD method performance}
Article \cite{rathinam_petzold_2003} discuss the computational advantages of applying POD to general linear and nonlinear dynamical systems. Figures~\ref{realgraph_rom} and \ref{randomgraph_rom} present the program execution times for ROM and regular models on real-world and random graphs, respectively.
Figures~1(a) and~2(a) show the execution times for evaluating the vector field in the regular and ROM cases, as defined in Equations~(\ref{Rd dynamics}) and~(\ref{fa}). In both datasets, ROM evaluation is consistently faster than the regular model for this step.
For the matrix--vector products---where the matrix represents the derivative of the forward dynamics with respect to either the states or the parameters, multiplied by a random vector (Figures~1(b),~2(b) and Figures~1(c),~2(c))---the regular model achieves lower execution times than ROM. Specifically, Figures~1(b) and~2(b) correspond to derivatives with respect to individual states, while Figures~1(c) and~2(c) correspond to derivatives with respect to parameters (see Section~\ref{jacobians}). Although ROM is slower in these matrix--vector operations, these steps contribute only a small fraction of the total computational cost.

The overall execution times, shown in Figures~1(d) and~2(d), combine the costs of vector field evaluation and both matrix--vector products. Despite ROM being slower in certain sub-operations, its advantage in vector field computation results in an overall much faster execution compared to the regular model.
\begin{figure}
\begin{subfigure}[t]{0.450\linewidth} \includegraphics[width=\linewidth,height=6cm]{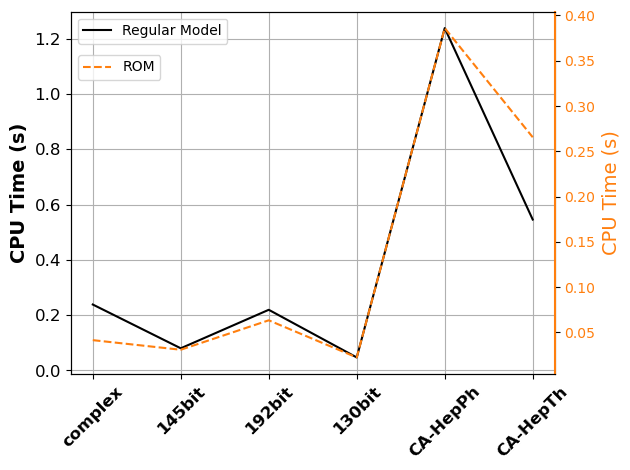}
\caption{Time taken ($t_{1}$) to evaluate vector field $f$ Equation (\ref{dynamicalPOD}) of regular model vs ROM vector field, see Equation (\ref{fa}).}
\end{subfigure}
\hspace{5mm}
\begin{subfigure}[t]{0.450\linewidth}
\includegraphics[width=\linewidth,height=6cm]{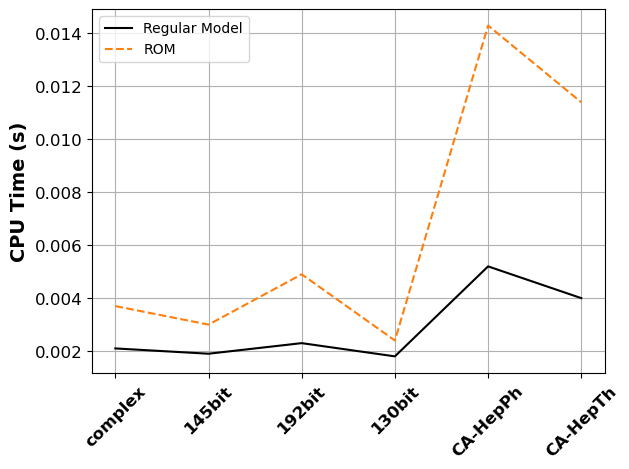}
\caption{Execution times ($t_{2}$) of the matrix-vector product with the matrix representing the derivative of forward dynamics with respect to the individual states for both the regular and ROM model (Section \ref{jacobians}).}
\end{subfigure}
\\
\begin{subfigure}[t]{0.450\linewidth}
\includegraphics[width=\linewidth,height=6cm]{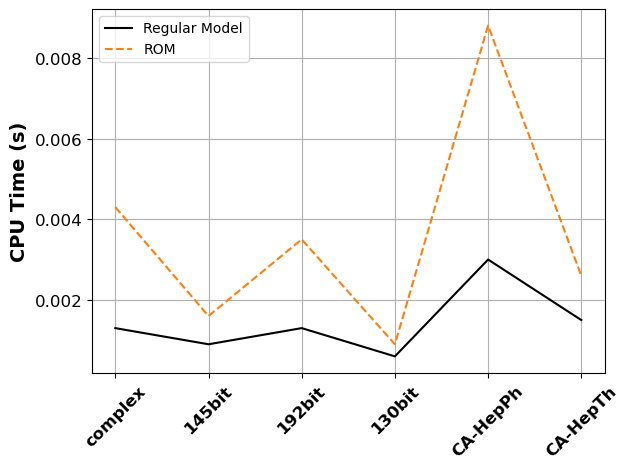}
\caption{Execution times ($t_3$) when evaluating the matrix-vector product with matrix representing the derivative of the forward dynamics with respect to the parameters (Section \ref{jacobians})}
\end{subfigure}
\hspace{5mm}
\begin{subfigure}[t]{0.450\linewidth} \includegraphics[width=\linewidth,height=6cm]{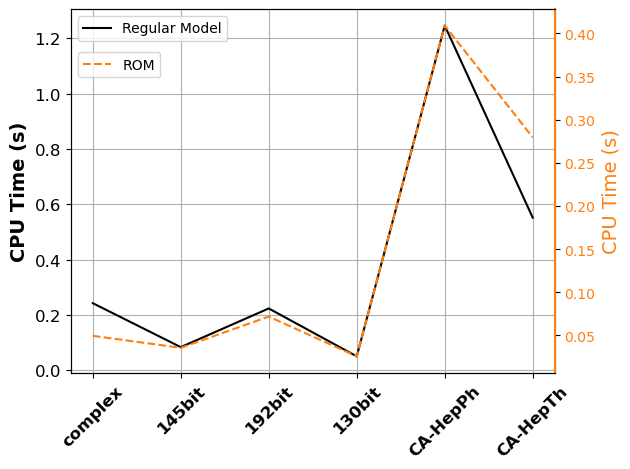}
\caption{ Overall time taken ($t_4 = t_1 + t_2 + t_3$) for ROM versus regular model.}
\end{subfigure}

\caption{Computational performance analysis of the chemical Brusselator model (Equation \ref{Rd dynamics}) represented on real-world graphs with and without POD.}
\label{realgraph_rom}
\end{figure}

\begin{figure}
\begin{subfigure}[t]{0.450\linewidth} \includegraphics[width=\linewidth,height=6cm]{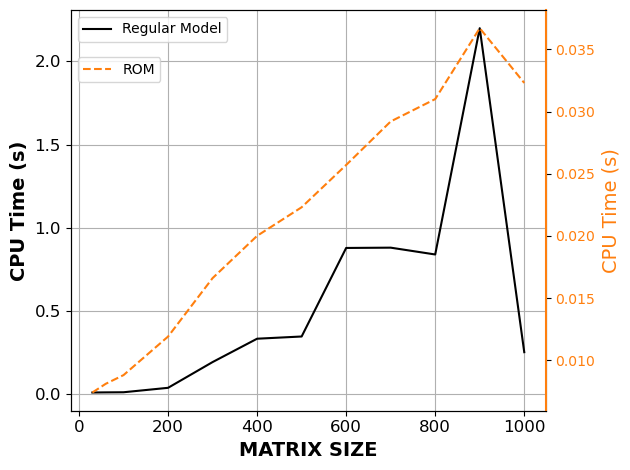}
\caption{Time taken ($t_1$) to evaluate vector field $f$ Equation \ref{dynamicalPOD} of regular model vs ROM vector field, see Equation (\ref{fa}).}
\end{subfigure}
\hspace{5mm}
\begin{subfigure}[t]{0.450\linewidth}
\includegraphics[width=\linewidth,height=6cm]{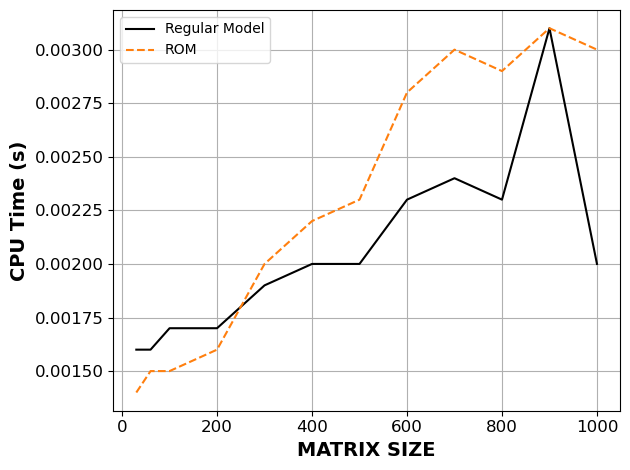}
\caption{Execution times ($t_2$) of the matrix-vector product with the matrix representing the derivative of forward dynamics with respect to the individual states for both the regular and ROM model (Section \ref{jacobians}).}
\end{subfigure}
\\
\begin{subfigure}[t]{0.450\linewidth}
\includegraphics[width=\linewidth,height=6cm]{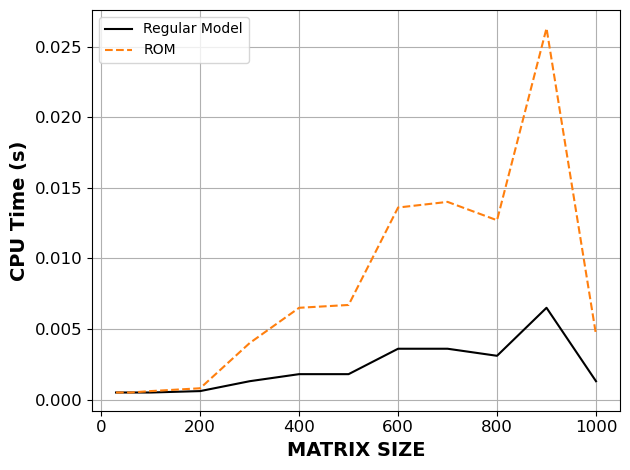}
\caption{Execution times ($t_3$) when evaluating the matrix-vector product with matrix representing the derivative of the forward dynamics with respect to the parameters (Section \ref{jacobians})}
\end{subfigure}
\hspace{7mm}
\begin{subfigure}[t]{0.450\linewidth} \includegraphics[width=\linewidth,height=6cm]{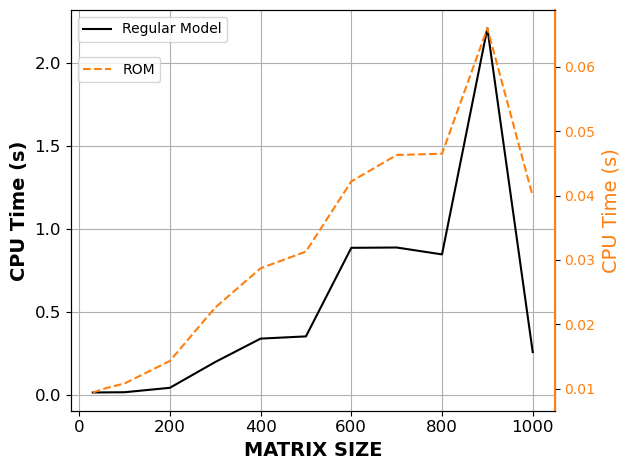}
\caption{ Time taken ($t_4 = t_1 + t_2 + t_3$) for ROM versus regular model.}
\end{subfigure}

\caption{Computational performance analysis of the chemical Brusselator model (Equation \ref{Rd dynamics}) represented on random graphs with and without POD.}
\label{randomgraph_rom}
\end{figure}




\section{Methodology Outline for Sparse Graph Construction with Trajectory-Based Pattern Preservation and Eigenmode Approximations}
\textcolor{black}{We present the \texttt{SGRDN} framework (Figure~\ref{datsapproach}) for efficient sparsification of reaction–diffusion systems on undirected graphs. The pipeline integrates simulation, POD-based model reduction, and an optimization problem with constraints on spectral properties, connectivity, and weight non-negativity, producing a sparse graph as the final output.}

Considering a subset of trajectories of a complex system on an undirected graph, the objective is to create an edge sparsifier (i.e. graphs with fewer edges than the original graph), and the sparse graph should produce patterns similar to the given subset of trajectories. The resultant sparse graph should adeptly mirror the behavioural patterns exhibited by the provided trajectory subset and also under random perturbations to the equilibrium point, {which is referred to as isodynamic behaviour}. This problem is posed as a constrained optimization problem where the objective function is made to include a term involving the difference between the observations from the given trajectories projected into a reduced dimension at various time points and the state vectors generated by the reduced order forward model based on the sparsified graph. An $\ell_1$ norm term of the weight multipliers ($\bar{\gamma}$) (Section \ref{multipliers_weight}) is also introduced as a penalty term in the objective function to induce sparsity. 

The new weight vector $\bar{w} = \textrm{diag($w$)} \bar{\gamma}$. We denote the perturbed Laplacian matrix $L^{'} = L + B^T E B, $ where $E = \text{diag(}\bar{w} -w )$. When this perturbed Laplacian matrix has many eigenmodes similar to the eigenmodes of the original network, the resulting patterns in the new graph will be highly correlated with the patterns in the original network. Thus any patterns or behaviours induced in the original network will likely be present in the new graph, see \cite{cencetti_clusella_fanelli_2018}. Computation of eigenmodes of the Laplacian matrix can be computationally intensive, so we propose approximations to determine the eigenmodes of the Laplacian matrix under perturbations, see Section \ref{estimate_eval}.

\subsection{Ensuring isodynamic behaviour}
This section discusses techniques to ensure isodynamic behaviour for dynamics on the sparse graph and the challenges involved. The dispersion relation connects the stability of the system as described by Equation (\ref{Rd dynamics}) with the Laplacian matrix eigenvalues. \textcolor{black}{The contribution of the $\alpha$-th eigenvalue of the Laplacian matrix ($\lambda_{\alpha}$) to the stability of the system (Equation~(\ref{Rd dynamics})) around the fixed point $r^{*}$ is given by the eigenvalues of the following matrix,}
\[S_{\alpha} = D_r\mathcal{F}(r^{\star}) - \left\lbrack \begin{array}{cc}
D_x \; & 0\\
0 & D_y 
\end{array}\right\rbrack \lambda_{\alpha} \; \]

\begin{figure}[H]
\centering
\includegraphics[width=10cm]{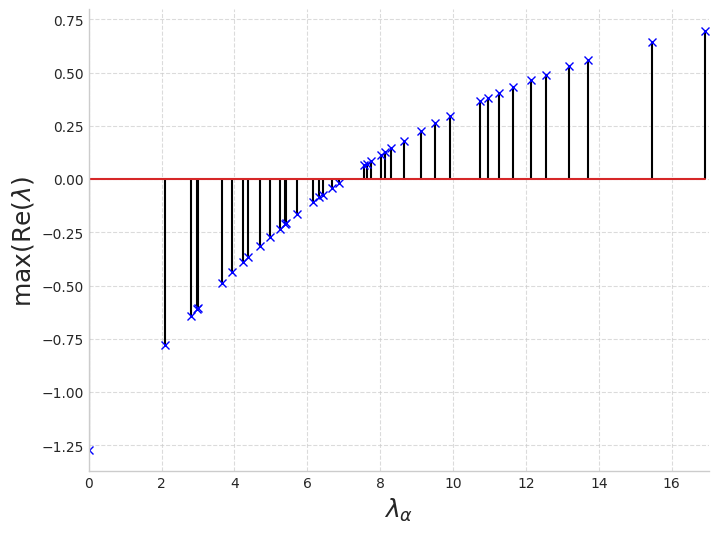}
\caption{Plot showing the stable and unstable eigenmodes of reaction-diffusion system as described by Equation (\ref{Rd dynamics}) on an Erd\H{o}s-R$\acute{e}$nyi random graph.}
\centering
\label{emodes_all}
\end{figure}
Figure~\ref{emodes_all} illustrates the distribution of stable and unstable eigenmodes for the reaction–diffusion system (Equation~\ref{Rd dynamics}) implemented on a random graph. \textcolor{black}{Figure~\ref{emodes_all} plots the maximum real part of the eigenvalues $\mathrm{Re}(\lambda)$ for each Laplacian mode $\lambda_{\alpha}$ of the reaction–diffusion system (Equation~\ref{Rd dynamics}). Modes below the red line are stable, while those above are unstable, with the transition point indicating the onset of pattern-forming instabilities.}

The study in \cite{cencetti_clusella_fanelli_2018} aims at producing a pattern invariant network using techniques like eigenmode randomization and local rewiring. One could use the error function employed in the local rewiring methodology as mentioned in \cite{cencetti_clusella_fanelli_2018} as a constraint to the optimization problem discussed in Section \ref{costfunc} but computing this function iteratively can become computationally expensive. The error function imposed as a constraint will be generally non-smooth. Methods to approach such non-smooth optimization problems using global optimization algorithms like Genetic algorithms are discussed in \cite{nocedalbook}. However, multiple function evaluations can make these methods ineffective. The error function ($\zeta$) introduced below is a modified variant of the one described in \cite{cencetti_clusella_fanelli_2018}.
\[
\zeta = n\zeta_l + \zeta_q
\label{errorfunc}
\]
\[
\zeta_l = \frac{\sum_{i = 1}^{n_p} |\Tilde{\lambda}_{i} - \lambda_{i}|^2}{n_p\;\; \sum_{i = 1}^{n_p} (\lambda_{i})^2}
\]
\[
\zeta_q = \frac{\sum_{i=1}^{n_p} |\Tilde{\phi}^T_{i} \phi_{i} - 1|^2 }{n_p}
\]
\textcolor{black}{Here $\lambda_i$ and $\phi_i$ denote the $i$-th eigenvalue and eigenvector of the original graph Laplacian, while $\Tilde{\lambda}_i$ and $\Tilde{\phi}_i$ represent the corresponding eigenvalue and eigenvector of the modified graph, $n_p$ denotes the number of eigenmodes preserved; we accept the new graph if the error function $\zeta$ is less than a tolerance value.}

In the presence of a connected input graph, we impose a minimum connectivity constraint to minimize the number of disconnected components, which, in turn, affects the error function (see Section \ref{errorfunctionestimate}). The intuition is taken from work in \cite{KUMAR2021100948}, which states that when the sum of the two lowest degrees of the graph component can influence the second smallest eigenvalue of the Laplacian matrix. We control the sum of degrees of the graph by ensuring that the unsigned incidence matrix times the weight vector is greater than a required connectivity level, say $\tau \mathbf{1}_{n\times 1}$ with $\tau$ being a user-defined parameter.
\textcolor{black}{We control the sum of degrees of the graph by ensuring that the product of the unsigned incidence matrix and the weight vector exceeds a required connectivity level, $\tau \mathbf{1}_{n\times 1}$, where $\tau$ is a user-defined parameter. This idea is illustrated in Figure~\ref{example_connect} for an unweighted graph with four nodes. Here, $Q$ denotes the unsigned incidence matrix, the weight vector is $w = \mathrm{diag}(\mathbf{1}_{5\times 1})\gamma$, with $\gamma$ being the edge multipliers (shown in red), and the degree vector is computed as $d = Q\gamma$.
}
\begin{figure}[htb]
\centering
\includegraphics[width=10cm]{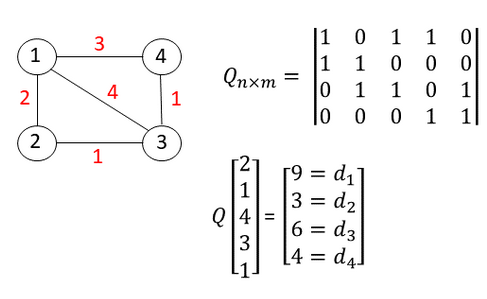}
\caption{{Example of imposing the connectivity constraint using the unsigned incidence matrix for an unweighted 4-node graph. Edge multipliers are shown in red.}
}
\centering
\label{example_connect}
\end{figure}

\subsubsection{Eigenvalue and eigenvector computation for the 
Laplacian matrix under perturbations}

\begin{lemma} \label{sherman}
(Sherman-Morrison formula). If A is a nonsingular $n \times n$ matrix and x is a vector, then
\[
(A + xx^T)^{-1} = A^{-1} - \frac{A^{-1}xx^TA^{-1}}{1+x^TA^{-1}x}.
\]
\end{lemma}
\begin{lemma}
(matrix determinant lemma) If A is nonsingular and x is a vector, then
\[
\text{det}(A + xx^T) = \text{det}(A)(1+x^TA^{-1}x).
\]
\label{MDL}
\end{lemma}

We present the challenges in finding the eigenpair during each iteration of the optimization procedure. If there is a change to only one of the edges of the graph, then one could use the approximate model described as in \cite{batson_spielman_srivastava_2014} making use of Lemma \ref{sherman} and \ref{MDL} to find the eigenpairs. At iteration $q$ of the algorithm, we will get a new graph with the weights modified according to the scale factor $\bar{\gamma}^q$. We define the matrix $E_1 = \text{diag(}\bar{\gamma}^q - \mathbf{1}_{m \times 1}$) and $\Tilde{\lambda} = \mathrm{diag(} \lambda_1, \lambda_2, \lambda_3,..,\lambda_n)$ denotes the eigenvalues of matrix $L$. We define the new Laplacian matrix as $L^{'} = L + B^T W^{1/2} E_1 W^{1/2}B.$ The characteristic polynomial of the matrix $L^{'}$ is given by
\[
\chi_{L^{'}}(\lambda, E_1) = \mathrm{det(} L^{'} - \lambda I) = \mathrm{det(} \Phi \, \Tilde{\lambda}\, \Phi^T + B^T W^{1/2} E_1 W^{1/2}B - \lambda I) \; 
\]
Using the property of determinants, we get the following.
\[
\chi_{L^{'}}(\lambda, E_1) = \text{det(} \Sigma + U^T E_1 U) \,\, \text{where}\,\, U^T = \Phi^T B^T W^{1/2}
\]
\[
\Sigma = diag(1, \lambda_2 - \lambda,...., \lambda_n - \lambda)
\]
\begin{align}
 \begin{array}{l}
\text{ $\chi_{L^{'}}(\lambda$, $E_1$) =  det}\left(\left\lbrack \begin{array}{cccc}
1 & 0 & \ldotp  & 0 \\
0 & u_2^T E_1 u_2 +\lambda_2 -\lambda \; & \ldotp  & u_2^T E_1 u_n   \\
\ldotp  & u_3^T E_1 u_2  & \ldotp  & u_3^T E_1 u_n   \\
\ldotp  & \ldotp  & \ldotp  & \ldotp   \\
0 & u_n^T E_1 u_2  & \ldotp  & u_n^T E_1 u_n + \lambda_n - \lambda  
\end{array}\right\rbrack \right)  \\
\\
\text{$\chi_{L^{'}}(\lambda$, $E_1$) = det} \left( \left\lbrack \begin{array}{ccc}
\mathrm{trace}\left(u_2 {\;u}_2^T \;E_1\right)+\lambda_2 -\lambda \; & \ldotp  & \ldotp\\
\mathrm{trace}\left(u_3 \;u_2^T \;E_1\right) & \ldotp  & \ldotp \\
\ldotp  & \ldotp   & \ldotp \\
\mathrm{trace}\left(u_n u_2^T \;E_1\right) & \ldotp  & \ldotp
\end{array}
\right\rbrack \right)
\label{proots}
\end{array}   
\end{align}
Finding the determinant of matrices with symbolic variables is described in \cite{https://doi.org/10.48550/arxiv.1304.4691}. We need to precompute only the diagonal elements of $u_i u^T_j$ to obtain the expression for the characteristic polynomial since $D$ is a diagonal matrix. Finding all the roots of the polynomial at each iteration and filtering out the highest roots may be more computationally expensive than finding only extreme eigenvalues as in (\cite{lehoucq_sorensen_yang_1998}, \cite{doi:10.1137/S0895479800371529}).
\\
\par
\textbf{Example:}
\begin{figure}[h]
\hspace{30mm}
\centering
\includegraphics[width=3cm]{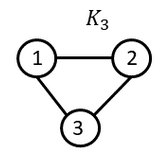}
\caption{Complete graph $\mathbb{K}_3$}
\label{k3graph}
\end{figure}
Let us consider the case of a complete graph on 3 vertices $\mathbb{K}_3$ Figure \ref{k3graph}. The Laplacian matrix of this graph is given by 
\[
\;\;\;L = 
\begin{bmatrix}
2 & -1 & -1 \\
-1 & 2 & -1 \\
-1 & -1 & 2 \\
\end{bmatrix}
\]
The eigenvalues of the matrix $L$ are $\lambda_1 = 0, \lambda_2 = 3, \lambda_3 = 3.$ Now let us perturb two of the weights of this graph and let the multiplier vector $\bar{\gamma} = [1.2, 1.4, 1]$. The matrix $D  = \mathrm{diag(\bar{\gamma}} - \mathbf{1}_{3\times1}).$ The matrix $B^T$ is given as follows,
\[
\;\;\;B^T = 
\begin{bmatrix}
1 & 0 & -1 \\
-1 & 1 & 0 \\
0 & -1 & 1 \\
\end{bmatrix}
\]

The expression given by Equation (\ref{proots}) will then become
$p(\lambda) = \mathrm{det(}  \begin{bmatrix}
3.9374 - \lambda & -0.0786 \\
-0.0786& 3.2626 - \lambda 
\end{bmatrix}  )$. $p(\lambda) = \lambda^2 - 7.2 \lambda + 12.84.$ The roots of the polynomial $p$ are given by 3.2536 and 3.9464, which are the eigenvalues of the matrix $L^{'}$.
\\
The effect on eigenvalues of a matrix with a perturbation is studied in \cite{bhatia_2007}. One such significant result in the field of perturbation theory is given below.
\begin{theorem}
(\cite{bhatia_2007}) If A and A + E are $n \times n$ symmetric matrices, then 
\[|\lambda_k (A + E) - \lambda_k (A) |\leq \vert \vert E \vert \vert_2\]
for $k = 1:n$
\end{theorem}
The aforementioned result can be employed to approximate the eigenvalue term in the error function ($\zeta_l$) described in Section \ref{errorfunc}. However, we currently lack a bound for the $\zeta_q$ term, and computing the $\ell_2$ norm of the matrix involves finding the largest eigenvalue, which is computationally expensive.
\\
In our endeavour to approximate the eigenvalues and eigenvectors of the transformed Laplacian matrix, we capitalize on the eigenvalues and eigenvectors of the original Laplacian matrix. This approach entails the utilization of the Rayleigh iteration method, which is further elucidated in the subsequent section. Section \ref{eigen_approx} provides an in-depth exploration of our approximation methodology, highlighting its distinctive advantages.\\
\textbf{Rayleigh Iteration:} $A \in \mathbb{R}^{n\times n}$ symmetric \\
$x_0$ given, $\vert \vert x_0 \vert \vert_2 = 1$
\begin{algorithmic}
\FOR{$i=0,1$ to ...}
\STATE $\mu_k = r_A(x_k)$ \, $(r_A(x) = \frac{x^T A x}{x^T x})$ 
\STATE Solve ($A$ - $\mu_k \mathbb{I})z_{k+1} = x_k $ for $z_{k+1}$ 
\STATE $ x_{k+1} = z_{k+1}/\vert\vert z_{k+1} \vert\vert_2 $
\ENDFOR
\end{algorithmic}

\subsubsection{Approximations to eigenpair}
\label{eigen_approx}
The Rayleigh quotient iteration algorithm is a powerful tool for computing eigenvectors and eigenvalues of symmetric or Hermitian matrices. It has the property of cubical convergence, but this only holds when the initial vector is chosen to be sufficiently close to one of the eigenvectors of the matrix being analyzed. In other words, the algorithm will converge very quickly if the initial vector is chosen carefully. 
From the first step in Rayleigh iteration, we get $\mu_0 = \lambda_{i} + \epsilon$, if we consider the $i$-th eigenvector $\phi_{i}$ as the initial point $x_0$. The next step requires solving for $z_1$ in the linear system $(A+E-\mu_0 I) z_1 = x_0$. Conducting the sensitivity of the linear system with a singularity is difficult to determine and is not trivial. The solution to the system $(A - \mu_0 I)z = \phi_{i}$ is given by $-\frac{\phi_{i}}{\epsilon}$. The first iterate using the conjugate gradient method would then be given by $y_{i} = \frac{1}{\epsilon \;}\;\left(\frac{a^T a\;\;E \phi_{i}}{a^T \left(A + E -\lambda_{i \;} I\;-\epsilon I\;\right)a\;\;}-\phi_{i \;} \right)$. We consider this iterate as our eigenvector approximation after normalizing. \textcolor{black}{ The proposed approximations to the eigenpair are given below, where $r_A(x) = \frac{x^T A x}{x^Tx}$ denotes the Rayleigh quotient:
\[(\hat{\lambda}_{i}, \hat{x}_{i}) = \bigg( r_{A+E} (y_{i}), \frac{y_{i}}{\vert\vert y_{i} \vert\vert} \bigg)\]
In the following theorem, we establish bounds on these eigenvalue and eigenvector approximations.}
\begin{theorem}
\label{theorem_estimate}

Let $G$ be an undirected graph with Laplacian matrix $L$ and $L^{'} = L + E$ be the perturbed Laplacian matrix, where $E = B^T D B$ represents the perturbations in the Laplacian matrix. We denote the $i$-th eigenvector estimate as
$\hat{x}_{i} = \frac{y_{i}}{\vert\vert y_{i} \vert \vert_2}$ with $y_{i} = \frac{1}{\epsilon \;}\;\left(\frac{a^T a\;\;E \phi_{i}}{a^T \left(A + E -\lambda_{i \;} I\;-\epsilon I\;\right)a\;\;}-\phi_{i \;} \right)$. Then 
\[\|(L + E - r_{L+E}(y_{i}) I) \hat{x}_{i}\|_2 \leq \tau + \|E\|_2,\]
where $r_A(x) = \frac{x^T A x}{x^T x}$, $\epsilon = \phi_{i}^T E \phi_{i}$, $a = E \phi_{i}$ and $\tau = max(\vert \lambda_n - \lambda_{2}(L+E) \vert, \lambda_n(L+E))$. $\lambda_n \geq \lambda_{n-1} \geq...\geq \lambda_2 \geq \lambda_1$ denotes the eigen values of the Laplacian matrix of graph $G$.
\end{theorem}
\begin{proof}

$\vert \vert ((L+E) - \frac{y_{i}^T (L+E)y_{i}}{y_{i}^T y_{i}}I ) \hat{x}_{i} \vert \vert = \vert \vert (\Phi \Lambda \Phi^T - \frac{y_{i}^T (L+E) y_{i}}{y_{i}^T y_{i}}I + E) \hat{x}_{i} \vert \vert $ (since Laplacian matrix is diagonalizable and $L = \Phi \Lambda \Phi^T$).
\\
\[\vert \vert (\Phi \Lambda \Phi^T - \frac{y_{i}^T (L+E) y_{i}}{y_{i}^T y_{i}}I + E) \hat{x}_{i} \vert \vert \leq  \]
\begin{align*}
\vert \vert \Phi (\Lambda - \frac{y_{i}^T (L+E) y_{i}}{y_{i}^T y_{i}}I) \Phi^T + E \vert \vert \,\,  \vert \vert \hat{x}_{i} \vert \vert = \vert \vert \Phi (\Lambda -  \frac{y_{i}^T (L+E) y_{i}}{y_{i}^T y_{i}}I) \Phi^T + E \vert \vert    
\end{align*}

Using the triangular inequality of matrix norms, the above inequality becomes 
\[ \leq 
\vert \vert \Phi (\Lambda - \frac{y_{i}^T (L+E) y_{i}}{y_{i}^T y_{i}}I) \Phi^T \vert \vert + \vert \vert E \vert 
\vert \leq  \tau  + \vert \vert E \vert \vert
\]
It can be observed that $\mathbf{1}^T y_{i} = 0$ and from the definition, we have $\lambda_k(L) = \mathrm{inf_{x \perp P_{k-1}} \,}\frac{x^T L x }{x^T x}$ where $P_{k-1}$ is the subspace generated by eigenvectors of the Laplacian matrix $L$ given by \\ $\mathrm{span(}v_1,v_2,...,v_{k-1}).$ Thus the Fielder value or the second smallest eigenvalue $\lambda_2$ of the Laplacian matrix is given by $\lambda_2 = \mathrm{inf_{x \perp \mathrm{1}}} \frac{x^T L x}{x^T x}$ and the largest eigenvalue $\lambda_n = \text{sup}_x \frac{x^T L x}{x^T x}$. Thus the above inequality becomes 
\[ \leq \tau  + \vert \vert E \vert \vert \mathrm{\,\,,where \,\, \tau = max(\vert \lambda_n - \lambda_{2}(\textit{L+E})\vert, \lambda_n(\textit{L+E}))}
\]
\end{proof}
The error function using the estimate would then be 
\[
\boxed{\bar{\zeta} = n \frac{\sum_{i=1}^{n_p} \vert \hat{\lambda}_{i} - \lambda_{i} \vert^2}{n_p \;\; \sum_{i = 1}^{n_p} (\lambda_{i})^2} + \frac{\sum_{i=1}^{n_p} \vert \hat{x}^T_{i} \phi_{i} - 1 \vert^2 }{n_p}}
\label{estimate_eval}
\]
From (\cite{reddy_trefethen_1990}), the following \textbf{definition} apply to pseudo spectra of matrices. \\
If $A$ is an $N \times N$ matrix and for $\lambda$ to be an $\epsilon-$pseudo-eigenvalue of A is 
\[ \sigma_{N}(\lambda I - A) \leq \epsilon\] where $\sigma_N$ denotes the smallest singular value. From the above definition, we can see that our estimate $\hat{\lambda}_{i}$ is a $\tau + \vert\vert E \vert\vert_2$-\textbf{pseudo eigenvalue} of \textit{L}. 
\\
We modify the constraint $\zeta \leq \beta_1$ by $\bar{\zeta} \leq k_1 \beta_1$, where $k_1$ and $\beta_1$ are user defined constants. \\
We present numerical experiments on certain graphs, and for every graph, we take the number of eigenmodes to preserve as the largest eigenmodes of the Laplacian matrix of the graph. Random perturbation is given to every edge of the graph, and a comparison is between $\zeta$ and $\bar{\zeta}$ (Table \ref{tab:my_label}).
\par Figures showing the effectiveness of the approximate eigenpair in preserving the eigenvalues and eigenvectors are described by Figures in Section \ref{image-eval1}. {Theorem \ref{Theo2}, which is a derived result, will play a crucial role in the subsequent discussion on the discontinuity of the function $\bar{\zeta}$ (refer to Eq. \ref{b}).} 
\newline
\begin{theorem}
Given an undirected graph $G = (V,E,W)$, with the Laplacian matrix $L$, then the matrix $M = L - (pI)$ is not orthogonal $\forall \,p \in \mathbb{R}$ if the minimum degree of the graph $d_{\mathrm{min}}(G) > 2.$
\label{Theo2}
\end{theorem}
\begin{proof}
We prove the above theorem by contradiction. For a matrix $M$ to be orthogonal, we require $M^T M = I$. The eigenvalue decomposition of the matrix $L = \Phi \Lambda \Phi^T$. Applying the definition of the orthogonality of matrices to the matrix $M$, we get $\Phi(\Lambda - pI)^2 \Phi^T = I$. Using the property of similarity of matrices, we can observe that all the eigenvalues of $\Phi(\Lambda - pI)^2 \Phi^T$ should be 1. From the zero eigenvalue of the matrix $L$, $p = \pm 1$. The remaining eigenvalues of the Laplacian matrix should be either 0 or 2 according to this $p$ value, but since the minimum degree of the graph is greater than 2, this cannot be the case as the largest eigenvalue of the Laplacian matrix of the graph will be greater than or equal to the minimum degree ($\lambda_n \geq d_{\text{min}} > 2,\; \lambda_n = \text{sup}_x \frac{x^T L x}{x^T x}$, taking $x = e_1$, where $e_i$ denotes the canonical vector with a 1 in the $i$th coordinate and $0'$s elsewhere).    
\end{proof}
\subsubsection*{Points of discontinuity for the function $\bar{\zeta}$}
Discontinuities in $\bar{\zeta}$ can happen due to the following cases, \\
\textbf{Case 1.}\, \[\phi_{i}^T E \phi_{i} = 0. \tag{d.(a)} \label{a}\] 
\\
\textbf{Case 2.}\, \[a^T(L + E - \lambda_{i}I - \epsilon I) a = 0, \text{\;where\;} a = E \phi_{i} \neq \mathbf{0}_{n\times1}. \tag{d.(b)} \label{b} \]  \\
Considering the first case, one possibility is when $ E $ is a rotation matrix. However matrix $E$ is not orthogonal $(\text{det}(E) = 0)$. The other possibility is when $\phi_{i}\, \in$ nullspace($E$). If this happens, then we could see that $(\lambda_{i},\phi_{i})$ is an eigenvalue, eigenvector pair for the new Laplacian matrix. Possible candidates of $\gamma$ satisfying this condition are given by the solution to the under-determined linear system under non-negativity constraints.

\[
 R \gamma = \lambda_{i} \phi_{i}, \, \mathrm{where\,\,} R = B^T W^{1/2} \text{$\text{diag}(W^{1/2}B\phi_{i})$}
\]
\[
\gamma \geq 0
\]
\\
Case 2 occurs when $ a \neq \mathbf{0}_{n\times1}\,\text{and}\, a^T ({L + E} - (\lambda_{i} + \phi_{i}^T E \phi_{i})I) a = 0.$ One possibility is that  $M_{i} = ({L + E} - (\lambda_{i} + \phi_{i}^T E \phi_{i})I)$ needs to be a rotation matrix, but this case cannot happen if we impose a minimum degree constraint as mentioned in Theorem \ref{Theo2}.
Elements in the nullspace of matrix $M_{i}$ under non-negativity constraints are candidates for discontinuity of the function $\bar{\zeta}$. \textcolor{black}{We enhance the eigenpair estimates of Theorem~\ref{theorem_estimate} by considering the impact of the following discontinuities.} \\
\[
\label{errorfunctionestimate}
(\tilde{\lambda}_{i}, \tilde{\phi}_{i}) = 
\begin{cases}
  (\lambda_{i}, \phi_{i}), \hspace{30mm} \text{if } E \phi_{i} = \mathbf{0}_{n \times 1}  \\ 
  (\frac{a^T M_{i} a}{a^T a}, \phi_{i}), \hspace{24mm}\text{if } y_{i} = \mathbf{0}_{n \times 1} \text{\;and\;} \\ \hspace{52mm}E \phi_{i} \neq \mathbf{0}_{n \times 1}\\
  (\lambda_{i} + \phi_{i}^T E \phi_{i}, \frac{a}{\vert\vert a \vert \vert}), \hspace{15mm} \text{if } M_{i} a = \mathbf{0}_{n\times 1}, \\ \hspace{45mm}y_{i}, E \phi_{i} \neq \mathbf{0}_{n\times 1} \\ \hspace{45mm}\text{\; and \;} \vert\vert y_{i} \vert\vert < \infty \\
  (\hat{\lambda}_{i}, \hat{x}_{i}), \hspace{30mm}  \text{otherwise} 
\end{cases}
\]
The error function is now modified as 
\begin{equation}
  \bar{\zeta}(\bar{\gamma}) = \frac{n}{n_p \sum_{i = 1}^{n_p} (\lambda_{i})^2} \sum_{i=1}^{n_p} (\tilde{\lambda}_{i} - \lambda_{i})^2 + \frac{1}{n_p} \sum_{i=1}^{n_p} (\tilde{\phi_{i}}^T \phi_{i} - 1)^2 .
  \label{zetatilde}
\end{equation}
The summation terms could be computed in parallel using multiple cores of the machine for large graphs. The speedup obtained when evaluating this function in parallel is shown in Table \ref{tab:my_label} and Figure \ref{speedupfigure}.


\textcolor{black}{
Table \ref{tab:my_label} demonstrates the computational efficiency and accuracy of the proposed eigenpair approximation method across various real-world graph datasets. The relative error between the exact error function and the approximated version remains consistently low across all tested graphs, ranging from 0.0013 for the complex graph to 0.0362 for the 192bit graph, validating the accuracy of our approximation scheme. The speedup columns compare the computational performance of our proposed method against traditional eigenvalue-eigenvector computation approaches. While smaller graphs (130bit, 145bit) show modest serial speedups of 2-3×, larger graphs demonstrate more substantial performance gains, with the Ca-HepTh dataset achieving an 83× speedup in serial execution. The parallel implementation further amplifies these benefits, particularly for larger graphs where the Ca-HepTh and Ca-HepPh datasets achieve speedups of 184× and 77× respectively. This performance scaling highlights the method's effectiveness for large-scale graph analysis, where the combination of approximation accuracy and parallel computational efficiency makes it particularly suitable for applications requiring rapid spectral analysis of perturbed graph Laplacians.
}

\subsubsection{Gradient of error function $\bar{\zeta}$}
\label{nonlinear_constr_grad}
\begin{align*}
\frac{d \bar{\zeta}}{d \gamma} = \frac{2n}{n_p \sum_{i = 1}^{n_p} (\lambda^{i})^2} \sum_{i = 1}^{n_p}  (\tilde{\lambda}_{i} - \lambda_{i}) \frac{d \tilde{\lambda}_{i}}{d \bar{\gamma}} + \frac{2}{n_p} \sum_{i = 1}^{n_p} (\tilde{\phi}^T_{i}, \phi_{i} - 1) \hspace{2mm} (D_{\bar{\gamma}} \tilde{\phi}_{i})^T \phi_{i}   
\end{align*}

\[
\hat{\lambda}_{i} = \frac{y(\bar{\gamma})^T \Tilde{L} y(\bar{\gamma})}{y(\bar{\gamma})^T y(\bar{\gamma})}
\]
\[
\frac{d \hat{\lambda}_{i}}{d \bar{\gamma}} = \frac{y(\bar{\gamma})^T y(\bar{\gamma}) \hspace{2mm} \nabla{y(\bar{\gamma})^T \Tilde{L} y(\bar{\gamma})} - (y(\bar{\gamma})^T \Tilde{L} y(\bar{\gamma})) \; \nabla{y(\bar{\gamma})^T y(\bar{\gamma})}}{(y(\bar{\gamma})^T y(\bar{\gamma}))^2}
\]

\[
\nabla_{\bar{\gamma}}{y(\bar{\gamma})^T \Tilde{L} y(\bar{\gamma})} = 2 D_{\bar{\gamma}} (B_1(\bar{\gamma}))^T B_1(\bar{\gamma}) 
\]
\[
B_1(\bar{\gamma}) = \bar{\gamma}^{1/2} W^{1/2} B y(\bar{\gamma})
\]
\[
\nabla{y(\bar{\gamma})^T y(\bar{\gamma})} = 2 (D_{\bar{\gamma}}y(\bar{\gamma}))^T y(\bar{\gamma})
\]
\[
\hat{x}_{i} = \frac{y(\bar{\gamma})}{\vert\vert y(\bar{\gamma})\vert\vert}, \;\nabla \left( \frac{y_j(\bar{\gamma})}{\vert\vert y \vert\vert}\right) = \left(\frac{\nabla y_j(\bar{\gamma})}{\vert\vert y \vert\vert}  - \frac{y_j(\bar{\gamma})\,D_{\bar{\gamma}}y(\bar{\gamma})^Ty(\bar{\gamma})}{\vert\vert y \vert\vert^3 }\right)
\]
\begin{table}[]
    \centering
    \begin{tabular}{ 
|p{2cm}|p{1cm}|p{0.8cm}|p{1.2cm}|p{1.5cm}|p{1.5cm}|p{1.4cm}|  }
\hline
\textbf{Graph} & n & $n_p$ & m & Relative error $\frac{(\zeta - \bar{\zeta})}{\zeta}$& speedup \;\; (serial) & speedup \;\;(parallel) \\
\hline
complex  & 1408 & 282 & 46326 & 0.0013
  & 1.5532 & 0.9061 
\\
\hline
145bit  & 1002 & 201 & 11251 & 0.0122 &  3.2414 & 0.4946 \\
\hline
192bit  & 2600 & 520 & 35880 & 0.0362 & 13.7567 & 11.9106 \\
\hline
130bit & 584 & 117 & 6058 & 0.0088 & 2.9379 & 0.1888 \\
\hline
Ca-HepPh & 11204 & 2241 & 117619 & 0.0079
 & 27.4243 & 77.1951\\
 \hline
 Ca-HepTh & 8638 & 1728 & 24806 & 0.0341 & 83.6472 & 184.8210 \\
\hline
\end{tabular}
    \caption{Showing comparison of $\zeta$ and $\bar{\zeta}$ for real world graphs with random perturbation to the edges.}
    \label{tab:my_label}
\end{table}
\label{estimatetable}

\section{Optimization problem}
\label{costfunc}
The sparsification framework we propose involves solving a constrained optimization problem, and the underlying constrained optimization problem is discussed in this section. In the optimization objective function (Eq. \ref{costfuncadjoint}), we strive to promote sparsity in the multipliers by incorporating a term that involves the $\ell_1$ norm. Optimization using $\ell_1$ term generally induces sparsity. The term ($T_1$) in the objective function seeks to minimize the disparity between the observations at different time points, projected onto a reduced dimensional space derived from various trajectories, and the states obtained at these times using the reduced order model with the weights $\bar{w} = \text{diag($w$)} \bar{\bar{\gamma}}$. Here $F^c(t_j)$ represents the observation from trajectory $c$ at time $t_j$. The initial conditions used to get observations $F^c(\cdot)$ will be kept invariant (Eq. \ref{9.c2}). $Q^{'} = Q \,\text{diag($w$)}$, where $Q$ represents the unsigned incidence matrix of the graph (Figure \ref{example_connect}), matrix $Q^{'}$ is used to impose the connectivity constraint (Eq. \ref{9.c3}). $A^c$ contains indices of observations assimilation from trajectory c. $\bar{F}^c(\cdot)$ follows the forward model $M(\cdot)$, the discretization for the reduced order model (Eq. \ref{9.c1}). Constraint $\bar{\zeta} \leq \beta_1$ tries to minimize perturbations to the largest $n_p$ eigenmodes of the Laplacian matrix (Eq. \ref{C}).  
\begin{equation}
\\
\\
\textbf{minimize} \,\, J(\bar{\gamma}) =  \underbrace{\frac{\sum_{c =1}^{\omega}
      \sum_{j \in A^c} \vert\vert(F^{c}(t_j )  -\bar{F}^{c}(t_j) \;)\vert\vert^2 }{2}}_{T_1}  + \color{blue}{ {\frac{\alpha}{2} \sum_{i=1}^m \vert \bar{\gamma}_i \vert}}     
\label{costfuncadjoint}    
\end{equation}
\begin{alignat*}{2}
\textbf{subject\,\, to} \\
& \tag{13.c1} \label{9.c1} \bar{{F}}^{c}_{q+1} = M(\bar{{F}}^{c}_{q},\bar{\gamma})\quad &&c = 1,\dots,\omega,\;\;q = 0,1,\dots,T-1,\\ &\bar{F}^c_0 = F^c_0  \tag{13.c2} \label{9.c2}\\
&Q^{'} \bar{\gamma} \geq \tau \quad && \tag{13.c3} \label{9.c3}\\
&\bar{\gamma}_i \geq 0, &&i = 1,2,\dots,m \tag{13.c4} \label{9.c4}\\
&\bar{\zeta}(\bar{\gamma}) \leq \beta_1 &&
\tag{13.c5} \label{C}
\end{alignat*}
 \par
 The first variation of $J$ is given by \[
 \delta J  = \sum_{c= 1}^{\omega} \sum_{j \in A_{c}} \left( (V_j^c)^T \eta_j^c\right)^T \delta\bar{\gamma} + \textcolor{black}{ \frac{\alpha_1}{2} \left(\frac{\vert \bar{\gamma}_1 \vert}{\bar{\gamma}_1},..., \frac{\vert \bar{\gamma}_i \vert}{\bar{\gamma}_i},..,\frac{\vert \bar{\gamma}_m \vert}{\bar{\gamma}_m}\right)^T \delta \bar{\gamma}}\,  
 \text{, see Equation (\ref{first_var})}
 \]
 
\begin{figure*}
    \centering
    \includegraphics[height = 0.6\textwidth, width=\textwidth]{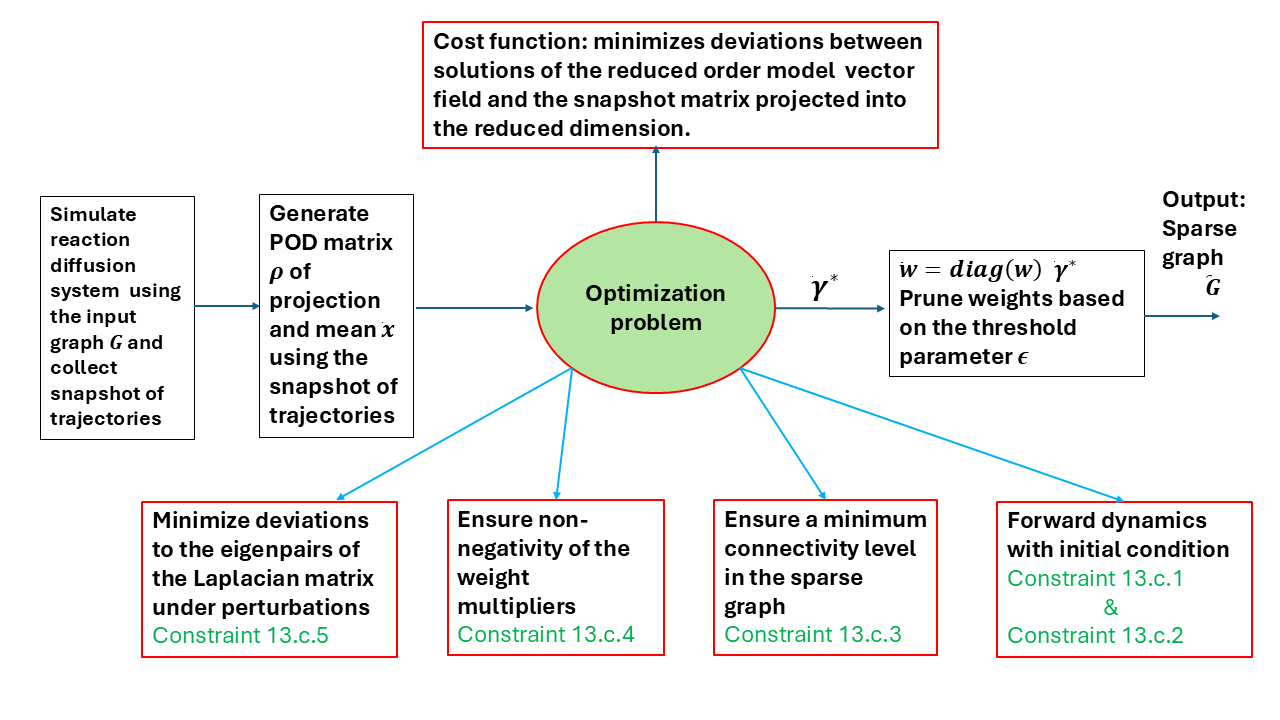}
    \caption{Schematic layout of the \texttt{SGRDN} algorithm.}
    \label{datsapproach}
\end{figure*}
The primary focus of the optimization routine is to minimize the function $J$ (Eq. (\ref{costfuncadjoint})) subject to constraints. The gradient of $J$ is obtained using the adjoint method algorithm as mentioned in (Section \ref{Adjointalgo}). The model $M(\cdot)$ used for forward dynamics represents the discretization for the ROM to reduce computational complexity. \textcolor{black}{The pruning of weights mentioned in the final step of Figure~\ref{datsapproach} is necessary because interior point methods, as discussed in \cite{nocedalbook}, typically converge to solutions in the vicinity of the optimum rather than achieving exact sparsity. The threshold-based pruning ensures that the final solution exhibits the desired sparse structure by setting near-zero weights to exactly zero.}
\par If we consider the optimization problem without the non-linear constraint, we can observe that the search space for our objective function described in (Eq. \ref{costfuncadjoint}) can be made closed by using the substitution ($\bar{\gamma} = \gamma^+ - \gamma^-$, $\vert\vert \bar{\gamma} \vert\vert_1 = \mathbf{1}^T \gamma^+  + \mathbf{1}^T \gamma^-$), where $\gamma^+$ represents the positive entries in $\gamma$, 0 elsewhere and $\gamma^-$ represents the negative entries in $\gamma$ and 0 elsewhere. The objective function and the constraints are continuous after this substitution. Hence, we encounter an optimization problem outlined as follows:
\begin{align*}
&\min_{x \in X}  f(x)\
&\text{subject to\;} h_j(x) \leq 0,\;\; j = 1,2,...n.
\end{align*}
Set $X$ is closed where functions $f$ and $h_j$ are continuous. We could use the Barrier methods for optimization as described in (\cite{nocedalbook},\cite{luenberger2021}). The feasibility set $S = X \cap \{x \,|\, h_j(x) \leq 0 \}$. If we generate a decreasing sequence $\epsilon_k$, let us denote $x_k$ as a solution to the problem $\mathrm{min_{x \in S}} \,\,f(x) + \epsilon_k B(x)$. 
Then every limit point of $\{x_k\}$ is a solution to the problem. One selection of barrier function $B(x) = -\sum_{j=1}^n \log(-h_j(x))$(Barrier convergence theorem). However, because of the discontinuities in the non-linear constraint, convergence in the constrained optimization problem may not be obtained. Heuristic algorithms for global optimization described in \cite{nocedalbook}, like Genetic algorithms, Pattern search algorithms exist however require multiple function evaluations.
\nolinenumbers
  Since the computation of the objective function in the optimization problem (Section \ref{costfunc}) is computationally expensive, we do not prefer such methods. 
  \\
\fbox{
\begin{minipage}{0.97\textwidth}
\textbf{SGRDN ALGORITHM} \\
\textbf{INPUT:} 
\begin{align*}
   &   \text{An undirected Graph\;}  G = (V,E,w) \\ 
                            & \alpha_1 - \; \ell_1 \; \text{regularization parameter} \\
                            & \tau_{\vert V \vert \times 1} - \text{connectivity parameter} \\
                            & \epsilon - \text{edge pruning parameter} \\
                            & N - \text{\;number of trajectories taken for POD step} \\
                            & \omega - \text{number of trajectories taken for assimilation} \\
                            & p_1, p_2,\dotsc,p_{\omega} - p_i \text{\;represents the number of points taken from trajectory $i$\;} \text{for assimilation}       \\
                            & k -\text{\;reduced dimension} \\
                            & n_p - \text{number of eigenmodes taken for the non-linear constraint} \\
                            &\beta_1 - \text{non-linear constraint parameter} \\
                            & F^1_0, F^2_0,....., F^N_0 - \text{initial conditions for $N$ trajectories} 
\end{align*}
\hrule height 1pt
\begin{algorithmic}[] 
\STATE 
\textbf{OUTPUT:} Sparse graph $\Tilde{G} = (V,\Tilde{E},\Tilde{w})$
\STATE \hrule height 1pt
\STATE \textbf{Procedure:}
\STATE 1. Run the forward reaction-diffusion model as in (Eq. \ref{Rd dynamics}) with initial conditions $F^1_0, F^2_0,....., F^N_0$ and sample $p_i$ number of data points in the interval $[0, T]$, $i = 1,2,\dotsc, \omega$ for assimilation, Sample data points at regular intervals for each trajectory thus obtaining the snapshot matrix containing $N$ trajectories. 
\STATE
\STATE 2. From the snapshot matrix, obtain the matrix $\rho$ of projection and mean $\bar{x}$ as mentioned in (Section \ref{ROMappro}), selecting the $k$ largest eigenvectors of the covariance matrix.
\STATE
\STATE 3. From the $N$ trajectories, select $\omega \leq N$ trajectories for the data assimilation correction part.
\STATE 4. Using the sampled trajectories, run the minimization procedure on the optimization problem discussed in (Section \ref{costfunc}) to obtain the weight multipliers $\bar{\gamma}$.
\STATE 5. On obtaining the multipliers, construct the weights $\bar{w}$ using the multiplier $\bar{\gamma}$, $\bar{w} = \text{diag($w$)} \bar{\gamma}$. Remove weights less than $\epsilon$ to obtain the new weights $w^{'}$. 
\STATE 6. Construct $\Tilde{G}$ using the weights $w^{'}$.
\end{algorithmic}
\textbf{END ALGORITHM}
 \end{minipage}
}
\begin{figure}
\centering
\includegraphics[height = 0.5\linewidth, width=0.7\linewidth]{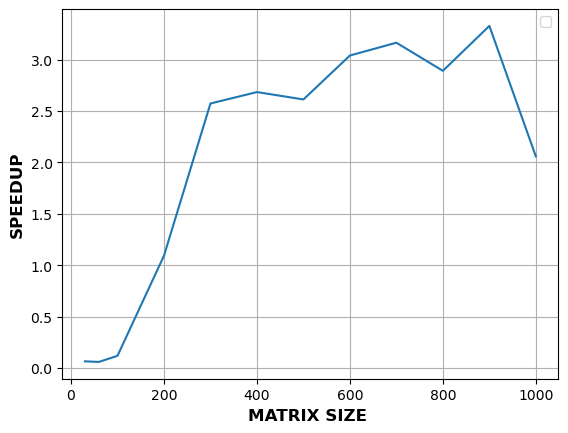}
\caption{Speedup obtained for non-linear constraint evaluation when using the approximations proposed in Section \ref{eigen_approx}.}
\label{speedupfigure}
\end{figure}
\subsection*{Running time of the procedure}
\label{timecomplexity}
\textcolor{black}{The time-consuming step in our sparsification procedure is solving a constrained non-linear programming problem (step 4 in the algorithm). For every iteration of the optimization procedure, we need to compute the projected vector field (Eq.~\ref{fa}), which takes $O(k|V|^2)$ operations (product of matrices $\rho$ and $L$). Evaluation of the gradient takes $O(k\vert E \vert)$ operations (Figure~\ref{costfunccomplexity}) from computation of the first term in Equation~\ref{first_var}.
The nonlinear constraint evaluation and the gradient evaluation take $O(|V|^2 n_p)$ ($n_p$ is the number of inner products between $\tilde{\phi}_i$ and $\phi_i$ in Eq.~\ref{zetatilde}) and $O(|E|^2 n_p)$ operations ($\nabla y(\bar{\gamma})^T \nabla y(\bar{\gamma})$ in Section~\ref{nonlinear_constr_grad}) as shown in Figure~\ref{nonlinearcomplexity}, respectively. Article~\cite{nocedal} describes the challenges involved in solving a nonlinear constrained optimization problem. We exploit parallelism in the nonlinear constraint evaluation steps, and the resulting speedups are shown in Figure~\ref{speedupfigure}.}

Figure~\ref{speedupfigure} illustrates the performance gains achieved when parallelizing the nonlinear constraint evaluation for different matrix sizes. The results show that speedups increase rapidly with matrix size, reaching a factor of over $3\times$ for larger matrices before slightly fluctuating at very high sizes. This demonstrates the effectiveness of our parallel implementation, especially for large-scale graphs where constraint evaluation dominates the runtime.


\begin{figure}[H]
\begin{subfigure}{0.5\textwidth}
\includegraphics[width=0.9\linewidth, height=5cm]{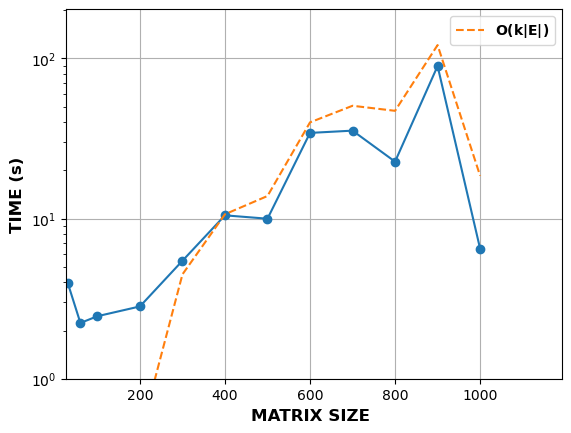} 
\caption{Figure showing the time complexity of cost function.}
\label{costfunccomplexity}
\end{subfigure}
\begin{subfigure}{0.5\textwidth}
\includegraphics[width=0.9\linewidth, height=5cm]{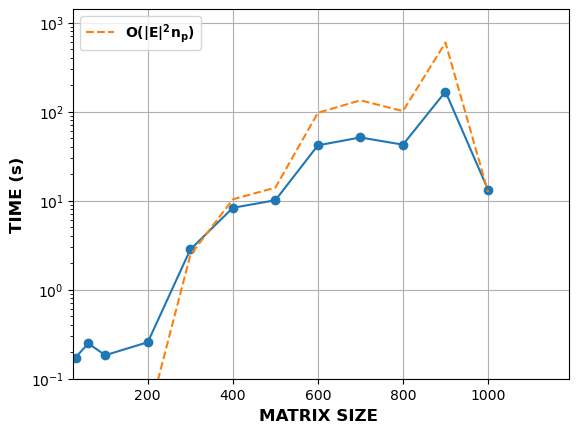}
\caption{Figure showing time complexity of the non-linear constraint.}
\label{nonlinearcomplexity}
\end{subfigure}

\caption{Time complexity of \texttt{SGRDN}.}
\end{figure}

\section{Results}
\label{results}
\texttt{SGRDN} has been applied to a series of graphs taken from actual case studies (\cite{nr-aaai15},\cite{snapnets}). The graphs under consideration are undirected, and preliminary processing steps are undertaken to eliminate self-loops. The parameters within the algorithm described in Section \ref{costfunc} were assigned as follows and were found after extensive parameter tuning and optimization; $\alpha_1 = 16 \log(\vert E \vert)$, $p_1, p_2 $ = 50, $D_x = 0.1$, $D_y = 1.4$, $N = 2$, $\omega = 2$, $\beta_1 = 10^{-2}$, $\tau = \text{max(2.1,\,}0.1\,d_{\mathrm{min}})\mathbf{1}_{n\times 1}\, \, $,\,$ \epsilon = 10^{-2}$, $k = \mathrm{min(}\lceil\frac{n}{5}\rceil,50)$ and $n_p = \lceil n / 5\rceil$, where $d_{\mathrm{min}}$ denotes the minimum degree of the graph. We include the $n_p$ largest eigenvalues and eigenvectors of the Laplacian matrix in the non-linear constraint evaluation step. The Euler forward method is used to find the solution at discrete time steps in the adjoint sensitivity step, and the explicit Range-Kutta 4 slope method is used to generate snapshot matrices required for the POD step. The tabulated results on real-world graphs are mentioned below (Table \ref{Allresults}). The constrained optimization step is solved in MATLAB \cite{MATLAB}, and we are using the interior point method for constrained optimization as described in \cite{nocedal}. \textcolor{black}{To quantitatively assess the preservation of dynamical behavior, we employ the correlation coefficient $R$ as our primary evaluation metric. This coefficient is computed by comparing the spatiotemporal solution profiles of the reaction-diffusion system (Eq. \ref{Rd dynamics}) between the original and sparsified graphs, where both networks are subjected to identical random perturbations from their respective equilibrium states. The correlation analysis captures the degree to which the sparsified graph maintains the essential dynamical characteristics of the original network under perturbative conditions.} \\ 
\textcolor{black}{Table \ref{Allresults} presents the performance of the \texttt{SGRDN} algorithm on five real-world network datasets, demonstrating its effectiveness in generating sparse graph representations while preserving dynamical behavior. The sparsification results vary significantly across different graph topologies: the algorithm achieves substantial edge reduction for graphs such as 145bit (from 11,251 to 7,767 edges, representing a 31\% reduction), 130bit (from 6,058 to 4,350 edges, representing a 28\% reduction) and approximately 8\% percent reduction for the ca-HepTh graph. The algorithm was unable to find a sparse set of multipiers for the 192bit graph, however the process lead to a non-sparse reweighting which resulted in a change of the eigenvalue and eigenvector spectrum as shown in Figure~\text{9.4(a), (b), (c)}. The correlation coefficients $R$, computed by comparing reaction-diffusion dynamics solutions between original and sparsified graphs under random perturbations, consistently exceed 0.85 across all datasets. This high correlation validates that our spectral-preserving sparsification approach successfully maintains the essential dynamical characteristics of the original networks. The varying degrees of sparsification across different graph types highlight the algorithm's adaptive nature, where the optimization process naturally determines the minimal edge set required to preserve the dominant eigenmodes within the specified tolerance bounds, rather than applying uniform compression ratios regardless of network structure.}

\textcolor{black}{
To validate the spectral preservation quality of our \texttt{SGRDN} algorithm, we conducted comprehensive ablation studies examining the eigenvalue and eigenvector characteristics of both original and sparsified graphs. Figure 9 presents the spectral analysis results across five real-world datasets, where each row corresponds to a different graph (Complex, 130bit, 145bit, 192bit, and Ca-HepTh). The leftmost panels (Figure 9.(1-5)(a)) display the relative eigenvalue changes $\frac{\lambda_i - \hat{\lambda}_i}{\lambda_i}$ between the original Laplacian matrix and its sparsified counterpart, demonstrating that our method successfully preserves the dominant eigenvalues with minimal relative error across all eigenmodes. The middle panels (Figure 9.(1-5)(b)) show the first two coordinates of all eigenvectors from the original graph's Laplacian matrix projected onto a 2D plane, revealing the structural complexity and distribution patterns of the spectral embedding space. The rightmost panels (Figure 9.(1-5)(c)) present the corresponding eigenvector coordinates for the sparsified graphs, where the visual similarity to the original eigenvector patterns confirms that our algorithm maintains the essential geometric structure of the spectral embedding. The radial symmetry and clustering patterns observed in both original and sparsified eigenvector plots validate that the proposed sparsification approach preserves not only individual eigenvalues but also the higher-order spectral relationships that are crucial for maintaining dynamical behavior on the reduced graph structures.
}

\textcolor{black}{We present an illustrative example of our framework's performance on a 50-node Erd\H{o}s-R$\acute{e}$nyi random graph, as depicted in Figure \ref{randomnode}. Figure 10(a) shows the original dense random graph structure with its characteristic high connectivity and complex edge patterns typical of Erd\H{o}s-R$\acute{e}$nyi networks. The \texttt{SGRDN} algorithm successfully identifies and preserves the most spectrally significant edges, resulting in the sparsified graph shown in Figure 10(b), which maintains a core-periphery structure while dramatically reducing the total edge count. To validate the dynamical preservation capabilities of our approach, we simulate Brusselator reaction-diffusion dynamics on both the original and sparsified networks under random perturbations to the equilibrium point. The spatiotemporal evolution patterns presented in Figures 10(c) and 10(d) demonstrate remarkable similarity between the dynamics on the original and sparsified graphs, with both exhibiting consistent wave propagation characteristics and temporal oscillatory behavior. The preserved pattern formation and synchronization dynamics across the reduced network topology confirm that our spectral-preserving sparsification approach successfully maintains the essential dynamical properties required for accurate simulation of complex networked systems, thereby validating the practical utility of the \texttt{SGRDN} framework for computational efficiency without sacrificing dynamical fidelity. \\}
\begin{table}[H]
\begin{align*}
\begin{tabular}{
|p{2cm}|p{2cm}|p{3cm}|p{2.5cm}|p{2.5cm}|p{2cm}|  }
\hline
\textbf{Graph} & Number of nodes & Number of edges in the given graph & Number of eigenmodes computed ($n_p$) & Number of edges in the sparsified graph & Correlation-coefficient ($R$) \\
\hline 
complex  & 1408 & 46326 & 282 & 39926 & 0.8654\\
\hline 
145bit  & 1002 & 11251 & 201 & 7767 & 0.8737 \\
\hline 
192bit & 2600  & 35880 & 520 & 35880  & 0.9189\\
\hline 
130bit  & 584 & 6058 & 117 & 4350  & 0.8572\\
\hline 
Ca-HepTh  & 8638  & 24806 & 1728 & 22709 & 0.9666
\\
\hline
\end{tabular}
\end{align*}
\caption{Illustration of the \texttt{SGRDN} algorithm on real world graphs.}
\label{Allresults}
\end{table}

 \label{image-eval1}

\begin{figure}[H]
    \caption*{\textbf{Complex}}
    \begin{subfigure}[t]{0.360\linewidth}
        \includegraphics[width=\linewidth,height=6cm]{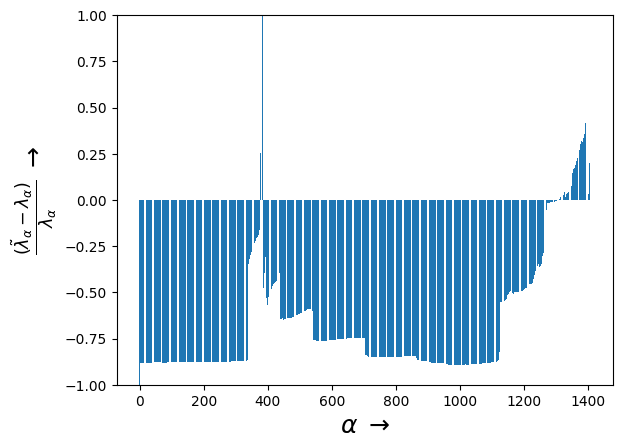}
        \caption*{9.1(a)}
    \end{subfigure}
    \begin{subfigure}[t]{0.320\linewidth}
        \includegraphics[width=\linewidth,height=6cm]{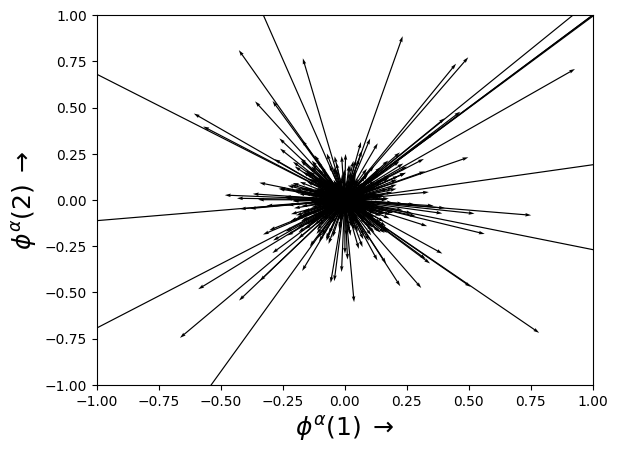}
        \caption*{9.1(b)}
    \end{subfigure}
    \hspace{-3mm}
    \begin{subfigure}[t]{0.320\linewidth}
        \includegraphics[width=\linewidth,height=6cm]{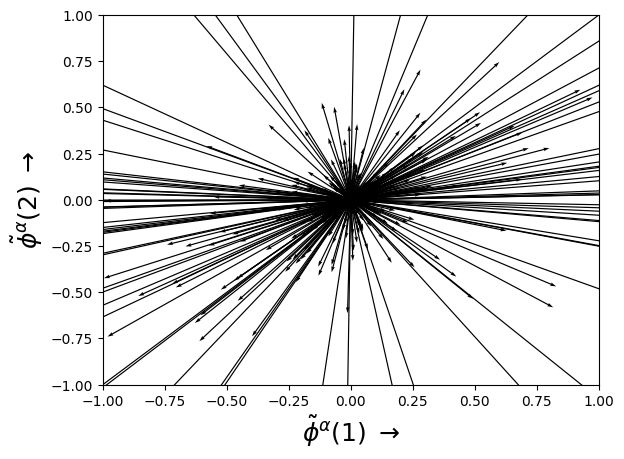}
        \caption*{9.1(c)}
    \end{subfigure}
\end{figure}

\begin{figure}[H]
    \caption*{\textbf{130bit}}
    \begin{subfigure}[t]{0.360\linewidth}
        \includegraphics[width=\linewidth,height=6cm]{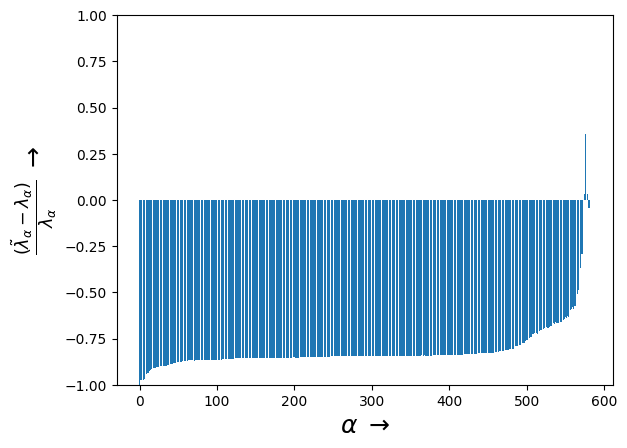}
        \caption*{9.2(a)}
    \end{subfigure}
    \begin{subfigure}[t]{0.320\linewidth}
        \includegraphics[width=\linewidth,height=6cm]{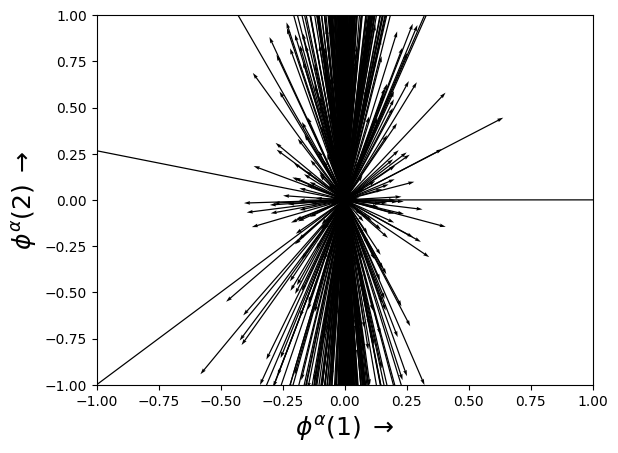}
        \caption*{9.2(b)}
    \end{subfigure}
    \hspace{-3mm}
    \begin{subfigure}[t]{0.320\linewidth}
        \includegraphics[width=\linewidth,height=6cm]{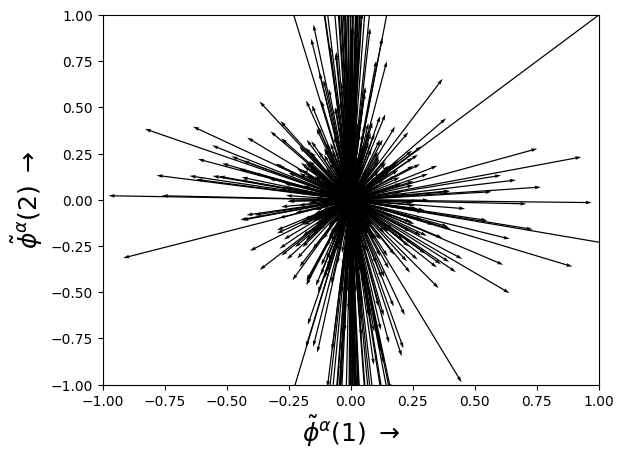}
        \caption*{9.2(c)}
    \end{subfigure}
\end{figure}

\begin{figure}[H]
    \caption*{\textbf{145bit}}
    \begin{subfigure}[t]{0.360\linewidth}
        \includegraphics[width=\linewidth,height=6cm]{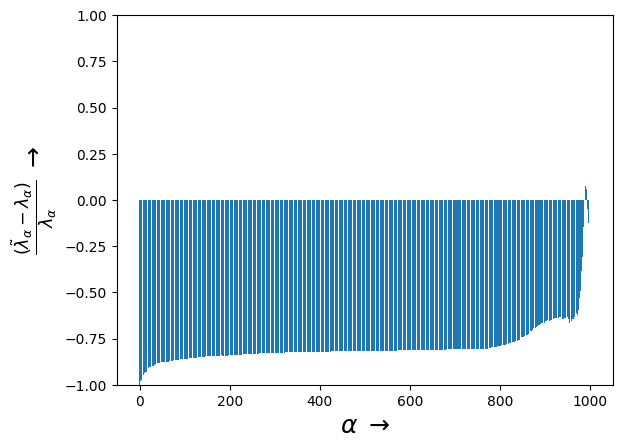}
        \caption*{9.3(a)}
    \end{subfigure}
    \begin{subfigure}[t]{0.320\linewidth}
        \includegraphics[width=\linewidth,height=6cm]{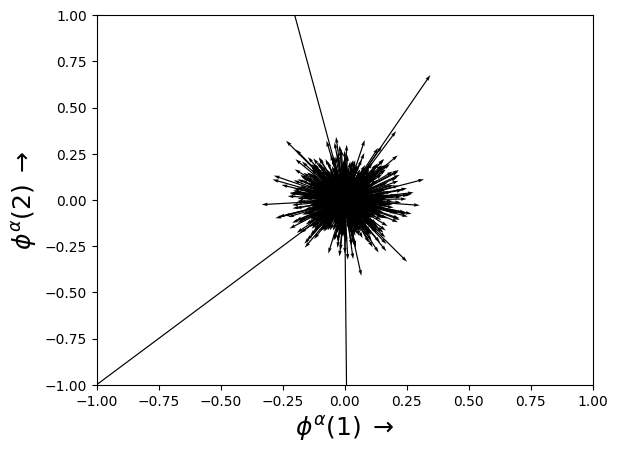}
        \caption*{9.3(b)}
    \end{subfigure}
    \hspace{-3mm}
    \begin{subfigure}[t]{0.320\linewidth}
        \includegraphics[width=\linewidth,height=6cm]{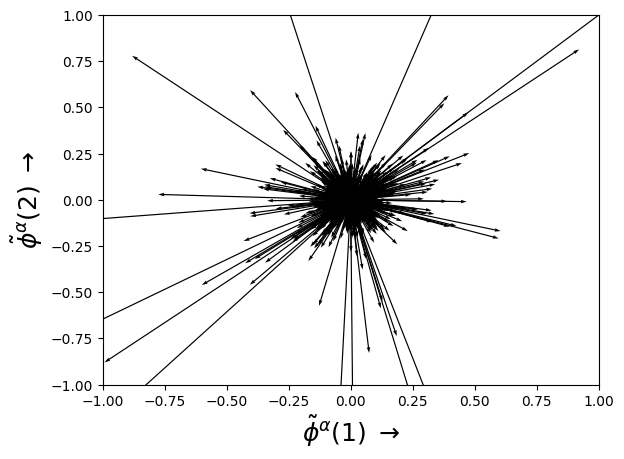}
        \caption*{9.3(c)}
    \end{subfigure}
\end{figure}

\begin{figure}[H]
    \caption*{\textbf{192bit}}
    \begin{subfigure}[t]{0.360\linewidth}
        \includegraphics[width=\linewidth,height=6cm]{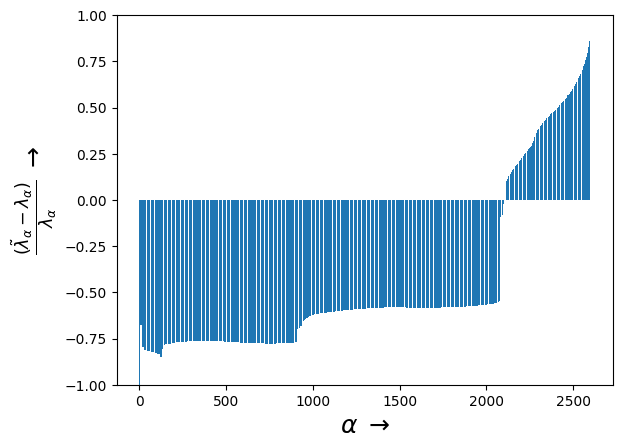}
        \caption*{9.4(a)}
    \end{subfigure}
    \begin{subfigure}[t]{0.320\linewidth}
        \includegraphics[width=\linewidth,height=6cm]{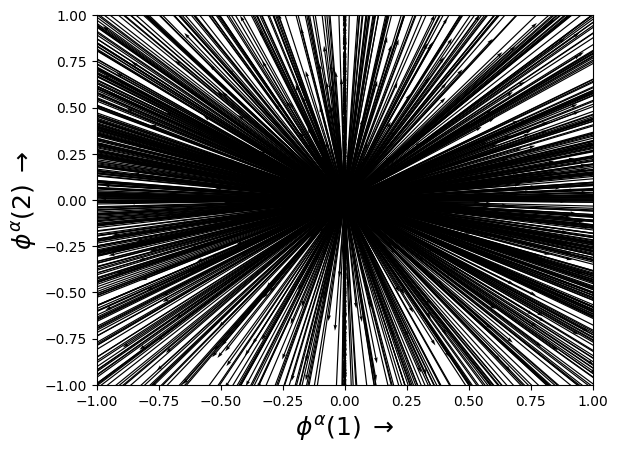}
        \caption*{9.4(b)}
    \end{subfigure}
    \hspace{-3mm}
    \begin{subfigure}[t]{0.320\linewidth}
        \includegraphics[width=\linewidth,height=6cm]{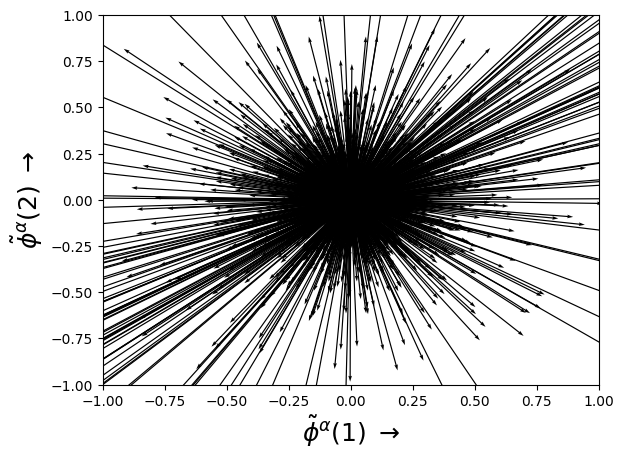}
        \caption*{9.4(c)}
    \end{subfigure}
\end{figure}
\addtocounter{figure}{-4}
\begin{figure}[H]
    \caption*{\textbf{Ca-HepTh}}
    \begin{subfigure}[t]{0.360\linewidth}
        \includegraphics[width=\linewidth,height=6cm]{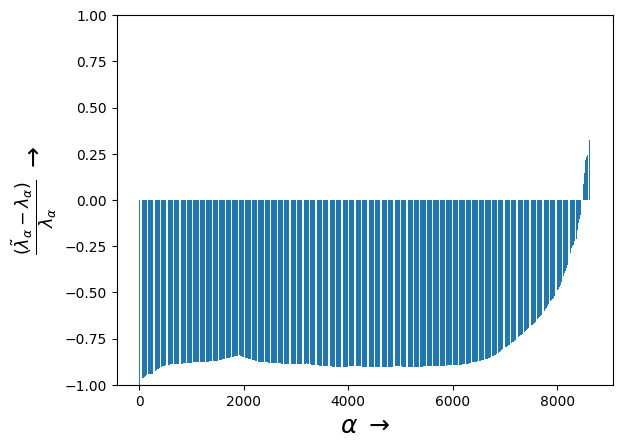}
        \caption*{9.5(a)}
    \end{subfigure}
    \begin{subfigure}[t]{0.320\linewidth}
        \includegraphics[width=\linewidth,height=6cm]{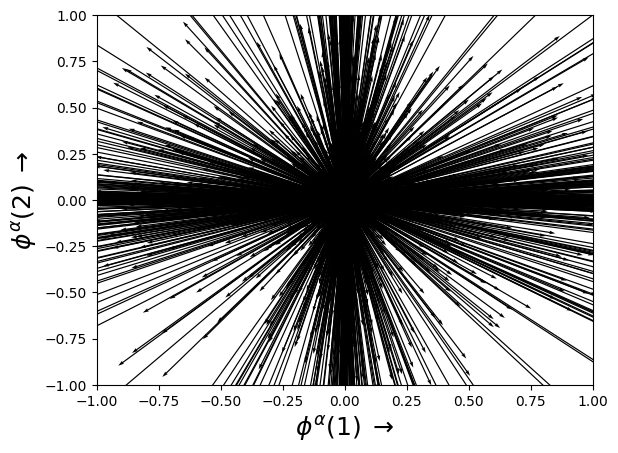}
        \caption*{9.5(b)}
    \end{subfigure}
    \hspace{-3mm}
    \begin{subfigure}[t]{0.320\linewidth}
        \includegraphics[width=\linewidth,height=6cm]{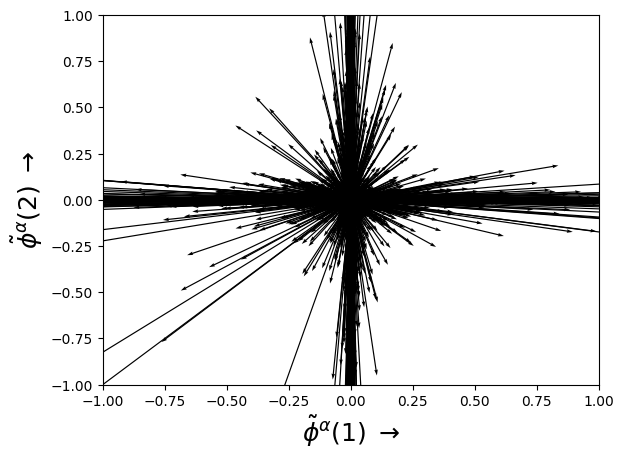}
        \caption*{9.5(c)}
    \end{subfigure}
     \caption{Spectral preservation analysis across five graph datasets. (a) Relative eigenvalue errors $(\hat{\lambda}_{\alpha} - \lambda_{\alpha})/\lambda_{\alpha}$ between original and sparsified graphs. (b) 2D visualization of original Laplacian eigenvectors. (c) 2D visualization of sparsified graph eigenvectors, demonstrating preservation of spectral structure.}
  \label{fig:overall}
\end{figure}
\begin{figure}[H]
    \centering
    \begin{subfigure}[t]{0.48\textwidth}
        \centering
        \includegraphics[width=\linewidth, height=6cm]{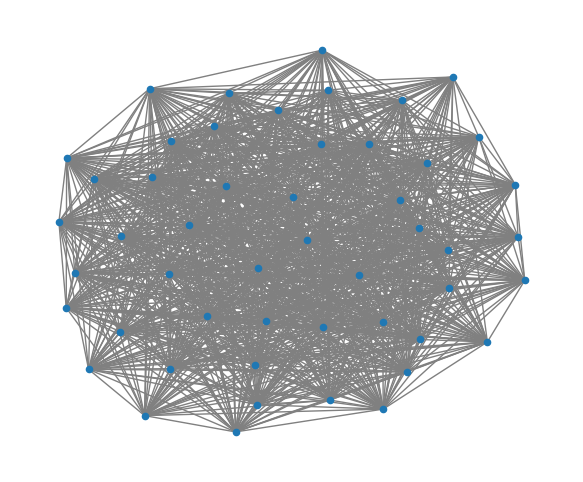}
        \caption{A 50-node Erd\H{o}s-R$\acute{\mathrm{e}}$nyi random graph.}
        \label{fig:original_graph}
    \end{subfigure}
    \hfill
    \begin{subfigure}[t]{0.48\textwidth}
        \centering
        \includegraphics[width=\linewidth, height=6cm]{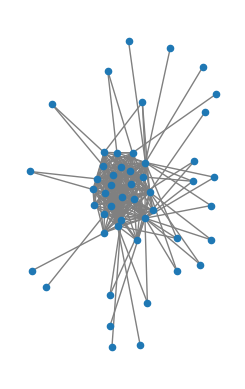}
        \caption{Sparsified graph obtained using the \texttt{SGRDN} algorithm.}
        \label{fig:sparse_graph}
    \end{subfigure}
    
    
    \begin{subfigure}[t]{0.48\textwidth}
        \centering
        \includegraphics[width=\linewidth, height=6cm]{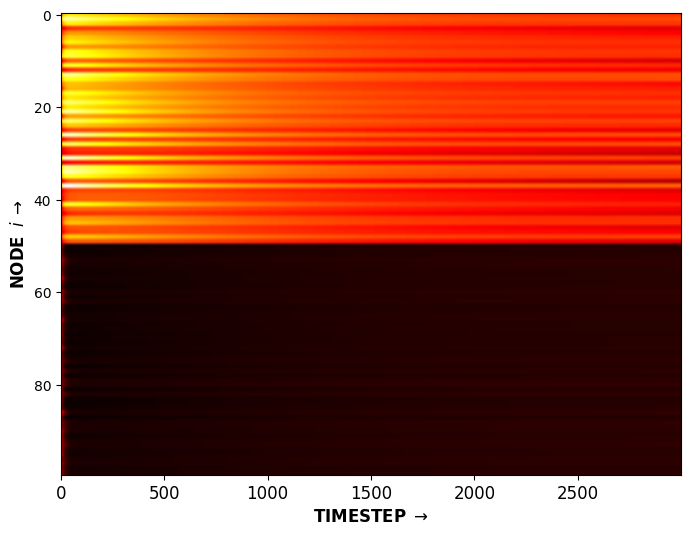}
        \caption{Brusselator reaction-diffusion dynamics on the original graph under random perturbations.}
        \label{fig:dynamics_original}
    \end{subfigure}
    \hfill
    \begin{subfigure}[t]{0.48\textwidth}
        \centering
        \includegraphics[width=\linewidth, height=6cm]{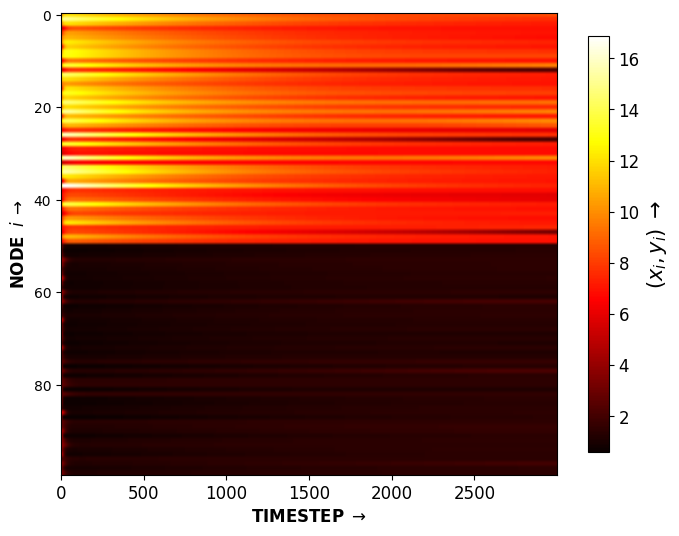}
        \caption{Brusselator reaction-diffusion dynamics on the sparsified graph under random perturbations.}
        \label{fig:dynamics_sparse}
    \end{subfigure}
    
    \caption{Performance demonstration of the \texttt{SGRDN} algorithm on a 50-node Erd\H{o}s-R$\acute{\mathrm{e}}$nyi random graph.}
    \label{randomnode}
\end{figure}
\section{Exploring Sparsity in ODENets: A Motivating Example}
\label{odenetsection}
This section presents a compelling illustration of how the framework can be adapted and applied to generate sparse neural ODENet models. Neural ODEs or Neural Ordinary Differential Equations were first introduced by Chen et al. in their 2018 paper "Neural Ordinary Differential Equations" \cite{DBLP:journals/corr/abs-1806-07366}.
The basic idea behind neural ODEs is to parameterize the derivative of the hidden state using a neural network. This allows the network to learn the system's evolution over time rather than just mapping inputs to outputs. In other words, the neural network becomes a continuous function that can be integrated over time to generate predictions.

One of the key advantages of neural ODEs is that they can be used to model systems with complex dynamics, such as those found in physics, biology, and finance. They have been used for various applications, including image recognition, time series prediction, and generative modelling. This section looks at the possibility of our framework in sparsifying such ODENets using \texttt{SGRDN}. In this section, our focus lies in addressing the following question: Can we mimic the dynamics of a dynamical system on a neural ODENet with a collection of snapshots of solutions using a sparse set of parameters by leveraging the adjoint sensitivity framework within a reduced-dimensional setting? \\
Let us consider an example of a linear dynamical system as defined in the following. 
\begin{equation}
\frac{dx}{dt} = Ax + b,\;\; A\in \mathbb{R}^{6\times6},\; b \in \mathbb{R}^6.
\label{LiDS}    
\end{equation}

For a predefined initial condition $x(0) = x_0$, the above problem will be an IVP. The neural network architecture proposed has one hidden layer with 50 neurons, with the input and output layers having six neurons for our experiment. A nonlinear activation is given at the output layer. The neural network output function proposed is given by \[nn(x(t)) = \sinh[\theta_2\theta_1x(t) + \theta_2 b_1 + b_2]\]
\[\text{\,\,} \theta_1 \in \mathbb{R}^{50\times6}, \theta_2 \in \mathbb{R}^{6\times50}, b_1 \in \mathbb{R}^{50}, b_2 \in \mathbb{R}^6.\]

\subsubsection*{Procedure}
\begin{enumerate}
    \item Generate the projection matrix $\rho \in \mathbb{R}^{2\times 6}$ and mean $\bar{x}\in\mathbb{R}^6$ of projection based on the POD method (Section \ref{ROMappro}) using the snapshot matrix consisting of data points from the linear dynamical system (Eq. \ref{LiDS}).
    \item Sample $p$ data points without replacement from the snapshot matrix for assimilation, denoting this data as $\mathcal{D} = \{ x_i\}_{i = 1,2,\ldots, p}$. Using the projection matrix, project the data into reduced dimension ($\mathcal{D}_{\text{proj}} = \{\rho(x_i - \bar{x})\}_{i = 1,2,\ldots, p}$). The projected neural network output function is given by,
    \[(nn_{\textbf{proj}}(z(t)) = \rho\, \sinh[\theta_2\theta_1\rho^Tz(t) + \theta_2\theta_1\bar{x} + \theta_2b_1 + b_2]\]
    \item For the sake of simplicity, we assume the parameters to be constant over time intervals. The adjoint sensitivity method for data assimilation is used to recover the parameters. The Euler forward method is used for the discretization of the projected neural ODE slope function. The cost function used for obtaining the parameters is given by (Eq. \ref{costfuncadjoint}). For this experiment, we utilized a set of 10 observations obtained from a single trajectory, randomly selected from a pool of 500 timesteps. Additionally, the chosen regularization parameter $\alpha$ was set to 40.

\begin{figure}[H]
  \centering
  \begin{minipage}[b]{0.5\linewidth}
    \centering
    \includegraphics[width=\linewidth, height = 6.2cm]{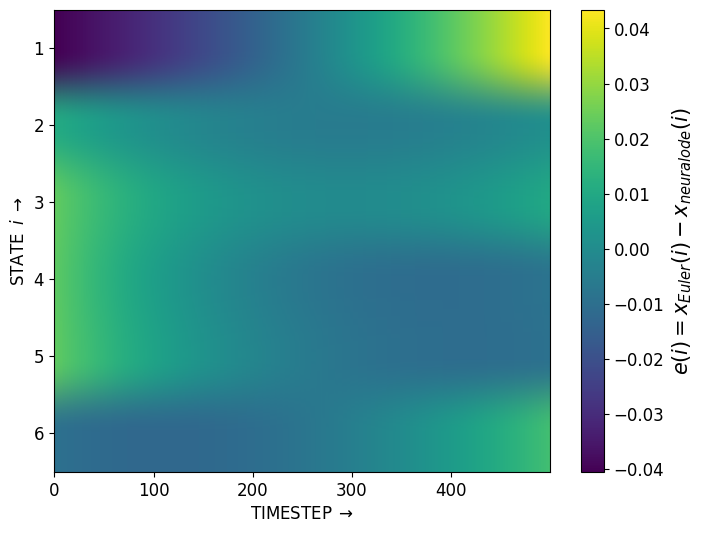}
    \caption{Error in the solution obtained from neural ODENet compared with the Euler forward method on the linear dynamical system (Eq. \ref{LiDS}).}
    \label{errornode}
  \end{minipage}
  \hfill
  \begin{minipage}[b]{0.40\linewidth}
    \centering
    \includegraphics[width=\linewidth, height =7cm]{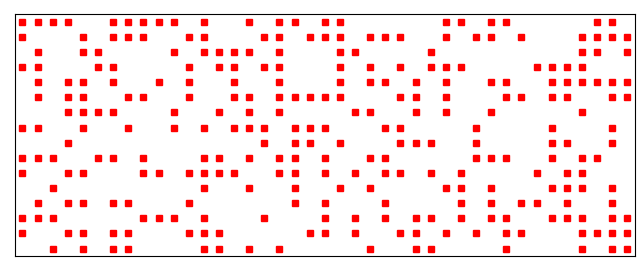}
    \caption{Sparsity pattern in the parameter vector of the neural ODENet.}
    \label{sparsitypattern}
  \end{minipage}
\end{figure}
\end{enumerate}
No constraints on the neural ODENet parameters are imposed, so the parameter estimation problem is unconstrained. \textcolor{black}{Figure \ref{errornode} illustrates the temporal evolution of the approximation error between the neural ODENet solution and the reference Euler forward method applied to the linear dynamical system (Equation \ref{LiDS}). The error surface reveals that the neural network successfully captures the underlying dynamics with minimal deviation, particularly in the early time steps, with errors remaining within acceptable bounds ($|\text{error}| < 0.04$) throughout the simulation period. 
Following the optimization process, the learned parameter matrix exhibits a distinct sparsity structure, as visualized in Figure \ref{sparsitypattern}, where red entries represent non-zero parameters and white entries correspond to parameters that have been effectively pruned to zero during training. Out of the total parameter space of 656 parameters, only 289 parameters remain non-sparse (56\% reduction), representing a significant reduction in model complexity while maintaining dynamical accuracy. This emergent sparsity pattern demonstrates that the neural ODENet naturally identifies the most critical connections required to represent the underlying system dynamics, providing both computational efficiency and interpretability advantages over dense parameter representations.} 

\section{Summary}
We present \texttt{SGRDN}, the first framework for the efficient sparsification of reaction-diffusion systems on undirected graphs (Figure \ref{datsapproach}). This framework formulates sparsification as a dynamic optimization problem operating on a subset of trajectories within a ROM space. The optimization objective maximizes the preservation of original patterns within the reduced dimension while simultaneously enforcing sparsity on the weight vector.

Our approach demonstrates strong performance, achieving significantly high correlation coefficients between the sparse and original graph solutions, even under random edge perturbations. Specifically, Table \ref{tab:my_label} validates the accuracy of our proposed eigenpair approximations (Theorem \ref{theorem_estimate}) against the actual eigenpairs for several real-world graphs. Furthermore, \texttt{SGRDN} proved effective in sparsifying ODENets, as evidenced by the resulting sparsity patterns (Figure \ref{sparsitypattern}).

To further enhance computational efficiency, we leveraged multi-core processing for evaluating the error function ($\bar{\zeta}$), resulting in a significant speedup (Figure \ref{speedupfigure}). The use of ROM-based corrections was critical in maintaining this strong computational performance across both random and real-world graphs (Figures \ref{randomgraph_rom} and \ref{realgraph_rom}).

\par \textbf{Scope and Constraints.} The \texttt{SGRDN} framework is generalizable to any reaction-diffusion system involving undirected graphs. However, sparsifying systems involving directed graphs, such as the heterogeneous Kuramoto model \cite{kuramoto_1984} and the Ginzgurg-Landau model \cite{garcía-morales_krischer_2012}, remains outside the scope of this work.

\section*{Acknowledgments}
The authors thank Prof. Dinesh Kumar, Prof. S. Lakshmivarahan, Prof. Arun Tangirala and anonymous reviewers for their valuable and encouraging comments, which helped produce an improved version of the work. This work was partially supported by the MATRIX grant MTR/2020/000186 of the Science and Engineering Research Board of India.
\bibliographystyle{plain}
\bibliography{bibfile}
\end{document}